\newif\ifpictures
\picturestrue

\newif\ifcomment
\commenttrue

\documentclass[11pt]{amsart}
\usepackage{latex_base}
\usepackage{macros}

\usepackage{indentfirst}
\usepackage{url}
\usepackage{natbib}
\usepackage{amsmath}
\usepackage{amsfonts}
\usepackage{xcolor, colortbl}
\usepackage{amsthm}
\usepackage{amssymb}
\usepackage{bm}
\usepackage{graphicx}
\usepackage{tikz}
\usetikzlibrary{shapes,snakes}
\usepackage{subfigure}
\usepackage{color}
\usepackage{rotating}
\usepackage{indentfirst}
\usepackage{multirow}
\usepackage{comment}
\makeatletter
\newif\if@restonecol
\makeatother

\usepackage[ruled,vlined,linesnumbered]{algorithm2e}

\def \cV { {\mathcal V} }
\def \pV { {\mathcal V} }
\def \xV { {\mathcal V} }

\def \cM { {\mathcal M} }

\def \cW { {\mathcal V} }
\def \cS { {\mathcal S} }
\def \dL { {\mathcal L} }
\def \dD { {\Omega} }
\def \coeff {{\rm coeff}}
\def \lcoeff {{\rm lcoeff}}
\def \lpp {{\rm lm}}

\newcommand{\LX}{{{\mathcal L_M}}}
\newcommand{\DD}{{{\mathcal D_M}}}
\newcommand{\su}{{{\mathcal S}\left(\Vector{u}\right)}}
\newcommand{\h}{{E_{\Vector{f}}}}
\newcommand{\B}{R}
\newcommand{\F}{E_{\Vector{F}}}

\DeclareMathOperator{\SmallAmbientVarietySU}{\widehat \cV(\cS)}
\DeclareMathOperator{\BigAmbientVarietySU}{\tilde \cV(\cS)}

\DeclareMathOperator{\proj}{proj}

\author{Xiaoxian Tang}
\address{Xiaoxian Tang, Department of Mathematics, Texas A\&M University,  College Station, TX 77843-3368, USA}
\email{xiaoxian@math.tamu.edu}

\author{Timo de Wolff}

\address{Timo de Wolff, Technische Universit\"at Berlin, Institut f\"ur Mathematik, Sekr. MA 6-2, Stra{\ss}e des 17.~Juni 136, 10623 Berlin,
 Germany\medskip}
 \email{dewolff@math.tu-berlin.de}

\author{Rukai Zhao}
\address{Rukai Zhao, Department of Computer Science \& Engineering, Texas A\&M University,  College Station, TX 77843, USA}
\email{zhaorukai@tamu.edu}



\title[]{Computing Elimination Ideals  and Discriminants of Likelihood Equations}

\begin{document}

\begin{abstract}
We develop a probabilistic algorithm for computing elimination ideals of likelihood equations, which is for larger models by far more efficient than 
directly computing Gr\"obner bases or the interpolation method proposed in \cite{RT2015, Tang2017}. 
The efficiency is improved by a theoretical result showing that the sum of data variables appears in most coefficients of the generator polynomial of elimination ideal. 
Furthermore, applying the known structures of Newton polytopes of discriminants, we can also efficiently deduce discriminants of the elimination ideals. For instance, the discriminants  of $3\times 3$ matrix model (Model \ref{ex:l6}) and one Jukes-Cantor model (Model \ref{ex:l9}) in phylogenetics (with sizes over $30$ GB and $8$ GB text files, respectively) can be computed by our methods. 
  \end{abstract}

 \maketitle

\section{Introduction}\label{sec:intro}


This work is motivated by the \struc{\textit{maximum likelihood estimation}} problem in statistics: 
\begin{quote}
	\textit{Which probability distribution describes a given data set optimally for a chosen statistical model?}
\end{quote}
A standard way to answer this question is to determine a point in the model that maximizes a \struc{\textit{likelihood function}}; see \eqref{eq:mle}. 
When the model is algebraic, see Definition \ref{definition:StatisticalModel}, and the data is discrete (i.e., a list of non-negative integers), then all critical points of the likelihood function can be found by solving a system of \struc{\textit{likelihood equations}} \eqref{eq:lle} via applying Lagrange multipliers. This motivates an important branch in algebraic statistics \cite{SAB2005, CHKS2006, BHR2007, HS2010,Uhler2012, GDP2012, EJ2014, HS2014, HRS,Rod14, ABBGHHNRS2017}. 
 
 Likelihood equations form an algebraic system in probability variables $p_0, \ldots, p_n$, Lagrange multipliers $\lambda_1, \ldots, \lambda_{s+1}$, and parameters $u_0, \ldots, u_n$ representing the data obtained from statistical experiments:
 \[f_0(u_0, \ldots, u_n; p_0, \ldots, p_n, \lambda_1, \ldots, \lambda_{s+1})= \cdots = f_{n+s+1}(u_0, \ldots, u_n; p_0, \ldots, p_n, \lambda_1, \ldots, \lambda_{s+1})=0.\] 
Given such a system with generically chosen data vector $(u_0, \ldots, u_n)$, the number of complex solutions is a finite non-negative constant, called the \struc{\textit{maximum-likelihood-degree (ML-degree)}}; see Definition \ref{def:MLDegree} and \cite{SAB2005, Huh13,  HS2014, BW15}. 

Since the variables $p_i$ represent probabilities, one is especially interested in a \struc{\textit{real solution classification}} \cite{BP2001, DV2005, CDMMX2010} of likelihood equations.
Unfortunately, this classification is very challenging, since it is a specific \struc{\textit{real quantifier elimination}} problem \cite{tarski1951, collins1975, arnon1988, mccallum1988, mccallum1999, grig88, hong1990a, hong1992, ch1991, renegar1992a, renegar1992b, renegar1992c, BPR1996, BPRRoadmap, BPRBook, brown2001a, brown2001b, brown2003, SS2003, SS2004, brown2012, HD2012, brown2013}, which is a fundamental problem in computational real algebraic geometry .  
 
 The number of real solutions only changes when the parameters (data) pass a set called the \struc{\textit {discriminant variety}}; see \citep[Definition 1]{DV2005} and Theorem \ref{th:th2xxt}. Hence, the discriminant varieties of likelihood equations, which is generated by homogenous polynomials \cite[Proposition 2]{Tang2017}, plays a core rule in real solution classification. We summarize the entire challenge in Figure \ref{fig:moti}. 

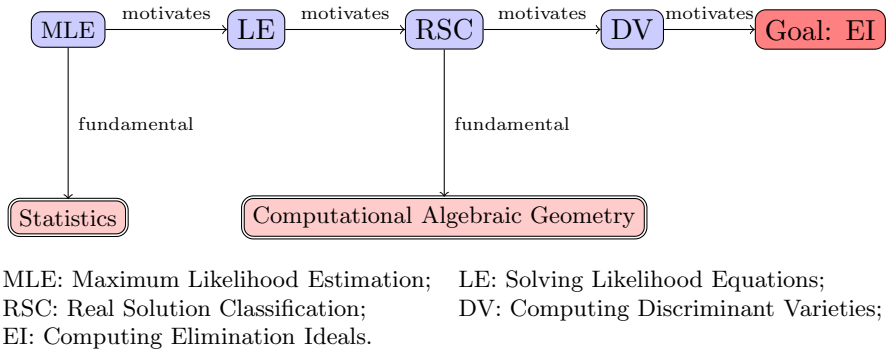
\begin{figure}[t]
\begin{tikzpicture}[node distance=2.5cm]
  \node [fill=blue!20,draw,  rounded corners]  (C) {\footnotesize MLE};
  \node [fill=red!20,draw, double, rounded corners] (F) [below of =C]{\footnotesize Statistics};
    \draw[->] (C) to node [right]{{\tiny fundamental}} (F);
  \node [fill=blue!20,draw,  rounded corners] (E) [right of = C]{LE};
  \node [fill=blue!20,draw,  rounded corners] (Bi) [right of=E] {RSC};
   \node [fill=red!20,draw, double, rounded corners] (H) [below of =Bi]{\footnotesize Computational Algebraic Geometry};
    \draw[->] (Bi) to node [right]{{\tiny fundamental}} (H);
  \node [fill=blue!20,draw,  rounded corners](Ai)[right of=Bi]{DV};
  \node [fill=red!50,draw,  rounded corners](Di)[right of =Ai]{Goal: EI};
  \draw[->] (C) to node [above]{{\tiny motivates}} (E);
    \draw[->] (E) to node [above]{{\tiny motivates}} (Bi);
   \draw[->] (Bi) to node [above]{{\tiny motivates}} (Ai);
    \draw[->] (Ai) to node [above]{{\tiny motivates}} (Di);
\end{tikzpicture}	
\\
\bigskip
{\footnotesize
\begin{tabular}{ll}
MLE: Maximum Likelihood Estimation; &
LE: Solving Likelihood Equations; \\
RSC: Real Solution Classification; &
DV: Computing Discriminant Varieties; \\
EI: Computing Elimination Ideals. &
\end{tabular}
}
\caption{An vizualization of the motivation for this work.}\label{fig:moti}
\end{figure}

In \cite{Tang2017}, Rodriguez and the first author studied how to compute discriminant variety, whose generator polynomial is called a \struc{\textit{data-discriminant}} in \cite[Definition 5]{Tang2017}, for a likelihood equation system efficiently. 
Experiments \cite[Tables 2--3]{Tang2017} suggest that the standard method \cite[section 3]{DV2005}, that is to compute Gr\"obner bases \cite{BB1965,faugere1993, faugere1999}, is not directly applicable for larger models. In \cite[Algorithm 2]{Tang2017} Rodriguez and the first author propose an probabilistic algorithm based on evaluation/interpolation techniques, which in theory works for arbitrary systems, in practice, however, are limited to small models with ML-degrees not greater than $6$.

\smallskip

The key idea of this article is to determine special structures of likelihood equations that help to improve the computational efficiency. In fact, likelihood equation systems are specific \struc{\textit{general zero-dimensional systems}}, see Definition \ref{def:gzd}, which are widely applied to problems in e.g., robotics \cite{CMABW2012, CJRM2014, CLO2015},  cellular differentiation\cite{cd2002, cd2005, cp2007, HTX2015}, and chemical reaction networks \cite{CF2005, CHW2008, HFC2013, JS2015, MFRCSD2016, CFMW2017, DMST2018}.

The solution set of a general zero-dimensional system of likelihood equations $f_0, \ldots, f_{n+s+1}$ can be represented by a triangular set $\{T_0, \ldots, T_{n+s+1}\}$ for generic choice of parameters $u_0, \ldots, u_n$; see Proposition \ref{shape}.
The discriminant variety of likelihood equations is a component of the discriminant locus of the univariate polynomial $T_0\in {\Q}(u_0, \ldots, u_n)[p_0]$ in this triangular system; see \cite[Lemma 3]{Tang2017}. Generally, it is sufficient to classify real roots of the univariate polynomial $T_0$ with respect to $p_0$ if we want to classify real solutions of the multivariate likelihood equations $f_0, \ldots, f_{n+s+1}$ with respect to $p_0, \ldots, p_{n}, \lambda_1, \ldots, \lambda_{s+1}$. 
 So 
 computing $T_0$ is crucial for real solution classification. In a general case, 
 $T_0$ generates the elimination ideal 
 \[\sqrt{\langle f_0, \ldots, f_{n+s+1}\rangle\cap {\mathbb Q}[u_0,\ldots, u_n, p_0]}.\]
Therefore, although one expects computing $T_0$ to be easier than computing a discriminant variety directly, it still means to compute a Gr\"obner basis of likelihood equations with respect to a  lexicographic  monomial order. Experiments show that it is a non-trivial task; see column ``Standard" in Table \ref{literatureOld}. 
 
Motivated by the discussion above, summarized in Figure \ref{fig:moti}, 
the {\bf goal} of this paper is to  efficiently compute the elimination ideal with respect to all parameters (data) and one variable for
a given system of likelihood equations by studying special structures of likelihood equations. More precisely, we have the {\bf problem statement}:
\\
{\bf ~~Input:} Likelihood equations $f_0, \ldots, f_{n+s+1}\in [u_0,\ldots, u_n, p_0, \ldots, p_{n}, \lambda_1, \ldots, \lambda_{s+1}]$; \\
{\bf ~~Output:} A generator of $\sqrt{\langle f_0, \ldots, f_{n+s+1}\rangle\cap {\mathbb Q}[u_0,\ldots, u_n, p_0]}$.

In this article, we achieve the following {\bf main contributions}:
\begin{enumerate}
\item We explore the structure of likelihood equations and prove the following main theorem: under some general enough hypotheses,  the sum of data $u_0 + \cdots + u_n$ appears as a factor with a particular power in the coefficients  of the generator polynomial of the elimination ideal of Lagrange likelihood equations; see Theorem \ref{th:main}. As a consequence, the sum of data is a factor of the discriminants of this generator polynomial; see Corollary\ref{cry:discr}.
\item Applying the main theorem to the interpolation method \cite[Algorithm 2]{Tang2017}, we develope a probabilistic algorithm, Algorithm \ref{interpolation}, for computing 
elimination ideals of Lagrange likelihood equations. Our experiments, which are summarized in Table \ref{literatureOld}, show that 
Algorithm \ref{interpolation} is significantly more efficient than the standard approach of directly computing Gr\"obner bases, or the evaluation/interpolation in \cite{Tang2017} for statistical models beyond very small size, see Table \ref{literatureOld} and Section \ref{sec:implementation} for further details. 

Applying the elimination ideals computed by Algorithm \ref{interpolation} and Corollary \ref{cry:discr}, we are in particular able to compute the discriminants of  $3\times 3$ matrix model (\ref{ex:l6}) and one Jukes-Cantor model (\ref{ex:l9}) in phylogenetics \cite[Chapter 15]{seth2018}, see Table \ref{comparediscr}, which was impossible before. We point out that these are gigantic polynomials, whose total degrees are $342$ and $176$, respectively, and which take several GB memory when stored in a text file; see Table \ref{comparediscr} for further details.
\end{enumerate}


\begin{table}[h]
\small
\centering
\label{literatureOld}
\begin{tabular}{|c|c|c|c|c|c|} \hline
\multirow{2}{*}{Models}&\multirow{2}{*}{$\# p_i$}&\multirow{2}{*}{ML-Degree}&
\multicolumn{3}{|c|}{Timings}\\
\cline{4-6} 
 &&&Standard& Interpolation &Algorithm \ref{interpolation}\\ \hline
  Model \ref{ex:l1} &4&3  &{\bf 0.046} s & 1.831 s &0.525 s\\\hline
  Model \ref{ex:l2} &6&2&{\bf 0.524} s &24.983 s & 2.310 s\\\hline
  Model \ref{ex:l3} &6&4& {\bf 3.211} s& 282.425 s& 16.174 s\\ \hline
  Model \ref{ex:l4} &6&6 & \textcolor{red}{$\infty$}&7933.230 s& {\bf 782.676} s\\ \hline
    Model \ref{ex:l5} &5&12 & \textcolor{red}{$\infty$} &10726.268  s& {\bf 761.257} s\\ \hline
     \cellcolor{blue!25}Model \ref{ex:l6} & \cellcolor{blue!25}9& \cellcolor{blue!25}10 & \cellcolor{blue!25}\textcolor{red}{$\infty$}  & \cellcolor{blue!25}\textcolor{red}{$>$ {\it 1583 d}} & \cellcolor{blue!25}{\bf 14} d\\ \hline
        Model \ref{ex:l7} &5&23 &\textcolor{red}{$\infty$}  &9919.260 s& {\bf 4624.575} s\\ \hline
          Model \ref{ex:l8} &8&14 &\textcolor{red}{$\infty$} & \textcolor{red}{$>$ {\it  4667 d}} & {\bf \textcolor{red}{$>$  {\it 15 d}}} \\ \hline
             \cellcolor{blue!25}Model \ref{ex:l9} & \cellcolor{blue!25}8& \cellcolor{blue!25}9 & \cellcolor{blue!25}\textcolor{red}{$\infty$} & \cellcolor{blue!25}\textcolor{red}{$>$ {\it 39 d}} &  \cellcolor{blue!25}{\bf 2} d\\ \hline
\end{tabular}
\smallskip 
\caption{Runtimes for computing elimination ideals (s: seconds; d: days).
The column ``standard'' constains the runtimes via a regular \texttt{FGb} Gr\"obner basis computation, the column ``Interpolation'' contains the runtimes for \cite[Algorithm 2]{Tang2017}, and the last column contains the runtimes for our Algorithm \ref{interpolation}.}
\end{table}

The article is organized as follows. Section \ref{sec:PrincipalEliminationIdealsliminaries} are preliminaries. We introduce the necessary notions and results from commutative algebra, on elimination ideals, algebraic statistics, and from the first author's previous paper \cite{Tang2017} on computing discriminant varieties  of likelihood equations. 
In Section \ref{sec:PrincipalEliminationIdeals}, we review/discuss the specialization properties of Gr\"obner bases, (radical) elimination ideals and multivariate factorization, and introduce  general zero-dimensional systems.
In Section \ref{sec:structure}, we prove the main results, Theorem \ref{th:main}, and Corollary \ref{cry:discr}.
In Section \ref{sec:alg}, based on the main theorem, we present and explain Algorithm \ref{interpolation} with a list of sub-algorithms for computing elimination ideals of likelihood equations. 
In Section \ref{sec:implementation}, we explain the implementation details and compare the efficiency of our code with existing tools. 
Also, we show how to compute discriminants more efficiently by the elimination ideals we have computed, and summarize the computational results for larger algebraic statistic models.

\subsection*{Acknowledgments}
We thank David A. Cox, Hoon Hong, Anne Shiu, and Frank Sottile for their support and advice. 
TdW was partially supported by the DFG grant WO 2206/1-1.


\section{Preliminaries}
\label{sec:PrincipalEliminationIdealsliminaries}

We assume that the reader is familiar with the fundamental concepts of computational algebraic geometry such as Gr\"obner bases and elimination ideals as well as related concepts in commutative algebra. For a general overview, we refer the reader to \cite{CLO2015} and \cite{Sturmfels:Book:SolvingSystemsofPolynomialEquations}.
\subsection{Notation}\label{sec:notation}
Throughout the paper, we use bold letters for vectors or a finite set of polynomials, e.g., $\struc{\Vector{z}}=(z_1,\ldots,z_n)$ and $\struc{\Vector{h}}=\{h_1,\ldots, h_m\}$. In any given vector space we denote the zero vector by $\struc{\Vector{0}}$.    For $h \in \Q[\Vector{z}]$ we denote the \struc{\textit{total degree}} of $h$ by $\struc{\deg(h)}$ and the degree of $f$ with respect to a particular variable $z_j$ as $\struc{\deg(h,z_j)}$.
We denote by $\struc{\coeff(h, z_j^i)}$  the \struc{\textit{coefficient}} of $h$ with respect to the monomial $z_j^i$.
If $N=\deg(h, z_j)$, then simply denote $\struc{\coeff(h, z_j^N)}$ by $\struc{\lcoeff(h, z_j)}$. 
For $\struc{\Vector{h}}\subseteq {\mathbb Q}[\Vector{z}]$, we denote by $\struc{\langle \Vector{h}\rangle}$ the \struc{\textit{ideal}} generated by $\Vector{h}$ in ${\mathbb Q}[\Vector{z}]$, and by
$\struc{{\mathcal V}(\Vector{h})}$ the \struc{\textit{affine variety}}
\[\struc{{\mathcal V}(\Vector{h})}~=~\{\Vector{z}\in {\mathbb C}^{n} \ | \ h(\Vector{z})=0,\; \forall h\in \Vector{h}\}.\]
For any ideal $\mathcal{I}\subset  {\mathbb Q}[\Vector{z}]$, we denote  by $\struc{\sqrt{\mathcal{I}}}$ the \struc{\textit{radical ideal}} of $\mathcal{I}$, and denote by $\struc{\mathcal{V}(\mathcal{I})}$ the affine variety defined by the generator polynomials of $\mathcal{I}$.
For any subset $\mathcal{S}\subseteq \C^{n}$, we denote by $\struc{{\mathcal I}(\mathcal{S})}$ the ideal generated by the polynomials vanishing on $\mathcal{S}$
\[\struc{{\mathcal I}(\mathcal{S})}~=~\{h\in {\mathbb Q}[\Vector{z}] \ | \ h(\Vector{z}^*)=0 \text{ for all } \Vector{z}^*\in \mathcal{S}\},\]
 and denote the \struc{\textit{Zariski closure}}  $\mathcal{V}(\mathcal{I}(\mathcal{S}))$ of ${\mathcal S}$ in $\C^{n}$ by $\struc{\overline{{\mathcal S}}}$. 
For a positive integer $n$, and for any $1\leq i \leq n$, we denote the \struc{\textit{canonical projection}} by 
\begin{align*}
	\struc{\proj_i}: {\mathbb C}^{n} ~\to~ {\mathbb C}^{i}, \quad (z_1, \ldots, z_{n}) \mapsto (z_1,\ldots,z_i).
\end{align*}

\subsection{Elimination Theory}
\label{subsec:EliminationTheory}

We recall two fundamental results from elimination theory, which will be frequently used in this article.
We denote by $\struc{\mathbb K}$ a field. 

\begin{proposition}\cite[page 121, Theorem 2]{CLO2015}\label{pro:elim}
Given $\Vector{h}\subseteq {\mathbb K}[z_1, \ldots, z_n]$, if 
$\mathcal{G}$ is a
Gr\"obner basis of $ \langle \Vector{h}\rangle$
with respect to the lexicographic order 
$z_{1}<\cdots <z_{n}$, then for any $1\leq i\leq n$, 
$\mathcal{G}\cap {\mathbb K}[z_1, \ldots, z_i]$ is a Gr\"obner basis of 
the elimination ideal $ \langle \Vector{h}\rangle \cap {\mathbb K}[z_1, \ldots, z_i]$.
\end{proposition}

\begin{proposition}\cite[page 131, Theorem 3]{CLO2015}\label{pro:clos}
Given $\Vector{h}\subseteq {\mathbb K}[z_1, \ldots, z_n]$, for any $1\leq i\leq n$, we have 
$\overline{\proj_i\left({\mathcal V}\left(\Vector{h}\right)\right)}~=~{\mathcal V}\left(\langle \Vector{h}\rangle \cap {\mathbb K}[z_1, \ldots, z_i]\right)$.
\end{proposition}

\subsection{Algebraic Statistics}
\label{subsec:AlgebraicStatistics}

In this section we recall the basic notions from algebraic statistics, which we need in this article.

\begin{definition}[\bf Probability Simplex]
\label{definition:ProbabilitySimplex}
We define the {\em \struc{$n$-dimensional probability simplex}} as 
$\struc{\Delta_{n}}~=~\{(p_0, \ldots, p_n)\in {\mathbb R}^{n+1}|p_0>0, \ldots,p_n>0, p_0+\cdots+p_n=1\}$. 
\end{definition}

With the probability simplex we define a fundamental object in algebraic statistics, the algebraic statistical model.

\begin{definition}[\bf Algebraic Statistical Model and Model Invariant]\label{definition:StatisticalModel}
Given homogenous polynomials $g_1,\ldots,g_s \in  {\mathbb Q}[p_0, \ldots,p_n]$ such that ${\mathcal V}(g_1, \ldots, g_s)\subsetneq{\mathbb C}^{n+1}$ is irreducible and generically reduced, we define an \struc{\textit{algebraic statistical
model}} as 
\begin{align*}
    \struc{\cM}~=~{\mathcal V}(g_1,\ldots,g_s)\cap \Delta_{n}.
\end{align*}
Each $g_i$ is called a \struc{\textit{model invariant}} of~$\cM$. 
If ${\mathcal V}(g_1, \ldots, g_s)$ has codimension $s$, then we say $\{g_1, \ldots, g_s\}$ is a set of \struc{\textit{independent model invariants}}. 
\end{definition}

\medskip

Given an algebraic statistical model $\cM$ and a \struc{\textit{data vector}} $\struc{\Vector{u}} = (u_0, \ldots, u_n)\in {\mathbb R}_{\geq 0}^{n+1}$, the \struc{\textit{maximum likelihood estimation} (MLE)} problem is the optimization problem
\begin{align}\label{eq:mle}
    \begin{aligned}
    	\max \ \Pi_{k=0}^np_k^{u_k} \ \text{ subject to } \Vector{p} \in \cM
    \end{aligned}
\end{align}
 which is fundamental in statistics \cite[Chapter 2]{DSS2009}. 
 In many sources, e.g., \cite{SAB2005, EJ2014, Tang2017}, an algebraic statistical model is defined by a projective variety generated by the model invariants $g_1, \ldots, g_s$. Since we are interested in the real critical points when solving the MLE problem, we prefer to work affinely and thus consider the affine cone over these projective varieties. 
One way to solve MLE problem is to solve a system of likelihood equations \citep*{SAB2005} formulated by Lagrange multiplier method.
We give the explicit formulation of such a system in what follows.
\begin{definition}[\bf Lagrange Likelihood Equations]\label{definition:LikelihoodEquations}
Given an algebraic statistical model $\cM$ with a set of independent model invariants $\{g_1, \ldots, g_s\} \subseteq  {\mathbb Q}[p_0, \ldots,p_n]$, the polynomial set $\Vector{f}=\{f_0, \ldots, f_{n+s+1}\}$ below is said to be the system of \struc{\textit{Lagrange likelihood equations}} of $\cM$ when set to zeros:  
\begin{align}\label{eq:lle}
\begin{array}{rl}
\struc{f_{0}(\Vector{u}, \Vector{p}, \Vector{\lambda})}~=&p_0(\lambda_1+\frac{\partial g_1}{\partial p_0}\lambda_2+\cdots+\frac{\partial g_s}{\partial p_0}\lambda_{s+1})-u_0,\\
&\quad\vdots\\
\struc{f_{n}(\Vector{u},\Vector{p}, \Vector{\lambda})}~=&p_n(\lambda_1+\frac{\partial g_1}{\partial p_n}\lambda_2+\cdots+\frac{\partial g_s}{\partial
p_n}\lambda_{s+1})-u_n,\\
\struc{f_{n+1}(\Vector{u},\Vector{p},\Vector{\lambda})}~=&g_1(p_0, \ldots,p_n),\\
&\quad\vdots\\
\struc{f_{n+s}(\Vector{u},\Vector{p},\Vector{\lambda})} ~=&g_s(p_0,\ldots,p_n),\\
\struc{f_{n+s+1}(\Vector{u},\Vector{p},\Vector{\lambda})}~=&p_0+\cdots+p_n-1,
\end{array}
\end{align}
where $\struc{\Vector{u}} = (u_0, \ldots, u_n)$, $\struc{\Vector{p}} = (p_0, \ldots,p_n)$, and $\struc{\Vector{\lambda}} = (\lambda_1,\ldots,\lambda_{s+1})$ are indeterminates.
More specifically,   $u_0,\ldots,u_n$ are parameters, and $p_0,\ldots,p_n, \lambda_1,\ldots,\lambda_{s+1}$ are variables.
\end{definition}
\begin{theorem}\citep*{SAB2005}\label{th:mld}
Given a system of  Lagrange likelihood
equations $f_0, \ldots, f_{n+s+1}$ defined in \eqref{eq:lle}, 
 there exist an affine variety $V\subsetneq {\mathbb C}^{n+1}$ and a non-negative integer $N$  such that for any $\Vector{b}\in {\mathbb
C}^{n+1}\backslash V$,  the equations $f_0(\Vector{b}, \Vector{p},\Vector{ \lambda})=\cdots=f_{n+s+1}(\Vector{b}, \Vector{p},\Vector{ \lambda})=0$ have $N$  common complex solutions in $\C^{n+1}\times \C^{s+1}$.
 \end{theorem}
 
The previous theorem motivates the following definition of the maximum-likelihood-degree.

\begin{definition}[\bf Maximum-Likelihod-Degree]\citep*{SAB2005}
\label{def:MLDegree}Given an algebraic statistical model $\cM$ with a system of Lagrange likelihood equations defined in \eqref{eq:lle}, the non-negative integer $\struc{N}$ stated in Theorem \ref{th:mld}
is called the \struc{\textit{maximum-likelihood-degree}}, short \struc{\textit{ML-degree}}, of $\cM$. 
\end{definition}

\begin{definition}\cite[Definition 4]{Tang2017}\label{def:nddv}
Given an algebraic statistical model $\cM$ with a system of Lagrange likelihood equations $\Vector{f}=\{f_0, \ldots, f_{n+s+1}\}$ defined in \eqref{eq:lle},
we define the following:
\begin{enumerate}
\item~$\struc{\LX_{J}}$ denotes
$\overline{\proj_{n+1}(\cV(\Vector{f})\cap {\mathcal V}(\struc{J}))}$, where 
$\struc{J}$ denotes 
 the determinant of Jacobian matrix of $\Vector{f}$ with respect to $(\Vector{p}, 
 \Vector{\lambda})$:
 {\footnotesize\[\det \left[
\begin{matrix}
\frac{\partial f_0}{\partial p_0} & \cdots & \frac{\partial
f_0}{\partial
p_n} & \frac{\partial f_{0}}{\partial \lambda_1} & \cdots & \frac{\partial f_{0}}{\partial \lambda_{s+1}}\\
\vdots & \ddots & \vdots &\vdots & \ddots & \vdots \\
\frac{\partial f_{n+s+1}}{\partial p_0} & \cdots & \frac{\partial
f_{n+s+1}}{\partial
p_n}& \frac{\partial f_{n+s+1}}{\partial \lambda_1} & \cdots & \frac{\partial f_{n+s+1}}{\partial \lambda_{s+1}}
\end{matrix}
\right].
\]}
	\item $\struc{\LX_{\infty}}$ denotes the set of the $\Vector{u}\in \overline{\proj_{n+1}(\cV(\Vector{f}))}$ such that there does not exist 
a   compact neighborhood $U$ of $\Vector{u}$ where
$\proj_{n+1}^{-1}(U)\cap \cV(\Vector{f})$ is  compact.

\end{enumerate}
\end{definition}

Both 
$\LX_{\infty}$ and $\LX_{J}$ are components of \struc{\textit{discriminant variety}}  \cite{DV2005} of Lagrange likelihood equations. 
Here, we interpret their geometry meanings, and roughly introduce how to compute them:
\begin{enumerate}
\item Geometrically, $\LX_{J}$ is the closure of the union of the projection of the singular
locus of $\cV(\Vector{f})$ and the set of
critical values of the restriction of $\proj_{n+1}$ to the regular locus of $\cV(\Vector{f})$ \cite[Definition 2]{DV2005}.
By Definition \ref{def:nddv} and Proposition \ref{pro:clos}, $\LX_{J}$ can be computed by computing the elimination ideal
$\langle \Vector{f}, J\rangle\cap \Q[\Vector{u}]$.

\item Geometrically,   $\LX_{\infty}$ is the set of parameters $\Vector{u}$ such that the Lagrange likelihood equations have some solution $(\Vector{p}, \Vector{\lambda})$ with coordinates tending to infinity. 
Also, $\LX_{\infty}$ is the closure of the set of non-properness of $\proj_{n+1}$ restricted on $\cV(\Vector{f})$ as defined in \cite[page 1]{Jelonek1999} and \cite[page 3]{SS2004}. By \cite[Lemma 2 and Theorem 2]{DV2005}, $\LX_{\infty}$ is an algebraically closed set and can be computed by Gr\"obner bases.  
\end{enumerate}
While the ML-degree captures the number of complex solutions of likelihood equations, 
$\LX_{\infty},$ and  $\LX_{J}$ (Definition \ref{def:nddv}) define open connected components such that the number of real solutions is uniform over each open connected component, see Theorem \ref{th:th2xxt}. 
In a more general setting, Theorem \ref{th:th2xxt} is a corollary of Ehresmann's theorem for which there exists semi-algebraic statements since 1992 \citep*{CS1992}.

\begin{theorem}\cite[Theorem 2]{Tang2017}\label{th:th2xxt}
Given an algebraic statistical model $\cM$ with a system of Lagrange likelihood equations $f_0, \ldots, f_{n+s+1}$ defined in \eqref{eq:lle},
if ${\mathcal O}$ is an open connected component~of 
${\mathbb R}^{n+1}\backslash (\LX_{J}\cup \LX_{\infty})$,
then for any  $\Vector{b}\in
{\mathcal O}$, 
the number of distinct real solutions of 
$f_0(\Vector{b}, \Vector{p},\Vector{ \lambda})=\cdots=f_{n+s+1}(\Vector{b}, \Vector{p},\Vector{ \lambda})=0$ in $\R^{n+1}\times \R^{s+1}$ is a constant.
\end{theorem}

Assume both ${\mathcal I}({\LX_J})$ and ${\mathcal I}({\LX_\infty})$ are principal, 
denote by $\struc{\DD_J}$ and $\struc{\DD_\infty}$  the generator polynomials of 
${\mathcal I}({\LX_J})$
and ${\mathcal I}({\LX_\infty})$, respectively. 
Notice that both $\DD_J$ and $\DD_\infty$ are
homogenous polynomials in $\Q[\Vector{u}]$, due to 
the structure of Lagrange likelihood equations \cite[Proposition 2]{Tang2017}. 

 According to Theorem \ref{th:th2xxt},  the product polynomial $\DD_{\infty}\cdot\DD_{J}$ plays a core rule in the real solution classification of Lagrange likelihood equations,  whose relations to 
 relevant concepts: \struc{\textit{border polynomial}} \citep*{BP2001}, and \struc{\textit{discriminant variety}}  \citep*{DV2005} are discussed in \cite[Remark 2]{Tang2017}.
 By the standard real solution classification method \cite{BP2001, CDMMX2010} used in {\tt Maple[RealRootClassification]}, we can obtain the parameter condition 
 under which the Lagrange likelihood equations have a certain number of common real solutions by two steps below. 
 \begin{enumerate}
\item[]{\bf Step 1.} Compute  $\DD_{\infty}\cdot\DD_{J}$.
\item[]{\bf Step 2.} Apply partial cylindrical algebraic decomposition to $\DD_{\infty}\cdot\DD_{J}$, and compute semi-algebraic descriptions of open connected components of 
 $\R^{n+1}\backslash (\LX_\infty\cup \LX_J)$. 
\end{enumerate}
So 
 computing $\DD_{\infty}\cdot\DD_{J}$ is crucial for classifying real solutions.  In practice, $\DD_{J}$ is much larger and more complicated to compute than $\DD_{\infty}$, because 
 $\DD_{\infty}$ can be obtained from a Gr\"obner base of $\langle \Vector{f}\rangle$ with respect to a graded monomial order, while  $\DD_{J}$ is computed from
 a Gr\"obner base of $\langle \Vector{f}, J\rangle$ with respect to a lexicographic order, where the determinant $J$ of Jacobian matrix defined as in Definition \ref{def:nddv} (2) can be  huge. 
 So here, we focus on computing $\DD_{J}$. 
 
In \cite{Tang2017}, the current methods with different strategies for computing $\DD_J$ are discussed in details (see \cite[Algorithms 1--2]{Tang2017}).  
Experiments \cite[Tables 2--3]{Tang2017} show that the largest model can be defeated so far is Model \ref{ex:l4} with ML-degree $6$. 
According to \cite[Algorithm 2, Strategy 3]{Tang2017},  one way to improve the efficiency for computing $\DD_{J}$ is to first compute the elimination ideal $\langle \Vector{f} \rangle\cap \Q[\Vector{u}, p_0]$,  instead of computing an elimination ideal $\langle \Vector{f}, J\rangle\cap \Q[\Vector{u}]$.
If  $\sqrt{\langle \Vector{f} \rangle\cap \Q[\Vector{u}, p_0]}=\langle \h\rangle$, 
 then it is well known that $\DD_{J}$ is a factor of 
the discriminant of $\h$ with respect to $p_0$ (for instance, one proof is  \cite[Lemma 3]{Tang2017}). However, $\h$ is not easy to obtained by directly computing Gr\"obner bases (see the column ``standard" in Table \ref{literatureOld}).  The goal of the rest of paper is to compute $\h$ more efficiently.

\section{Specialization Properties and General Zero-dimensional Systems}\label{sec:PrincipalEliminationIdeals}

In this section we discuss a selection of specialization properties and general zero-dimensional systems, which are both necessary for our theoretical results in Section \ref{sec:structure} and Algorithm \ref{interpolation} in Section \ref{sec:alg}.

\subsection{Specialization Properties}
\label{subsec:PrincipalEliminationIdeals}

In what follows, we consider polynomial rings with at least two variables, i.e., ${\mathbb Q}[z_1, \ldots,  z_n]$ with $n\geq 2$. Given $h\in {\mathbb Q}[z_1, \ldots,  z_n]$, we denote for every
$1\leq i<  n$, 
by  $\struc{\lpp_i(h)}$ and $\struc{{\lcoeff}_i(h)}$ the \struc{\textit{leading monomial}} and \struc{\textit{leading coefficient}} of $h$ \struc{\textit{with respect to} $z_{i+1}, \ldots, z_{n}$},  when $h$ is considered in $\mathbb{Q}(z_{1}, \ldots, z_{i})[ z_{i+1}, \ldots,  z_n]$  with the lexicographic order $z_{i+1}<\cdots<z_{n}$. For every $\Vector{b}=(b_1, \ldots, b_i)\in \C^i$, we define the polynomial
\begin{align*}
	\struc{h(\Vector{b})} \ = \ h|_{z_1=b_1, \ldots, z_i=b_i}\in \C[z_{i+1}, \ldots, z_n].
\end{align*}
For every polynomial set $\Vector{h}\subseteq {\mathbb Q}[z_1, \ldots,  z_n]$, we define
\begin{align*}
	\struc{\Vector{h}(\Vector{b})} \ = \ \{h(\Vector{b})\in \C[z_{i+1}, \ldots, z_n] \ | \ h\in \Vector{h}\}.
\end{align*}
 
 \begin{definition}\cite[Definition 4.1]{KSW2010}\label{def:noncomp}
Given  $\Vector{h}\subseteq {\mathbb Q}[z_1, \ldots,  z_n]$, for any 
$1\leq i< n$,   a subset $\Vector{g}$ of $\Vector{h}$ is a
\struc{\em noncomparable subset} of $\Vector{h}$  \struc{\textit{with respect to} $z_{i+1}, \ldots, z_{n}$} if 
\begin{enumerate}
	\item for every $h\in \Vector{h}$, there exists a $g\in \Vector{g}$ such that $\lpp_i(h)$ is a multiple of $\lpp_i(g)$, and
	\item for every $g_1, g_2\in \Vector{g}$, with $g_1 \neq g_2$, the leading monomial 
	$\lpp_i(g_1)$ is not a multiple of $\lpp_i(g_2)$, and 
$\lpp_i(g_2)$ is not a multiple of $\lpp_i(g_1)$. 
\end{enumerate}
\end{definition}

\begin{proposition}\label{pro:noncomp}\cite[Theorem 4.3]{KSW2010}
Given  $\Vector{h}\subseteq {\mathbb Q}[z_1, \ldots,  z_n]$, let ${\mathcal G}$ be a Gr\"obner basis of 
 $\langle \Vector{h}\rangle$ with respect to the lexicographic order $z_{1}<\cdots<z_{n}$.
 Fix an integer $i$ with $1\leq i< n$. 
Let ${\mathcal G}_{i}={\mathcal G}\cap {\mathbb Q}[z_1, \ldots, z_i]$, 
and let ${\mathcal N}$ be a noncomparable subset of ${\mathcal G}\backslash {\mathcal G}_{i}$ with respect to $z_{i+1}, \ldots, z_{n}$. Then for any $\Vector{b}\in \cV\left({\mathcal G}_{i}(\Vector{b})\right)\backslash \cV\left(\Pi_{g \in {\mathcal N}}{\lcoeff}_{i}(g)\right)$,  ${\mathcal N}(\Vector{b})$ is a Gr\"obner basis of $\langle \Vector{h}(\Vector{b})\rangle$  with respect to the lexicographic order $z_{i+1}<\cdots<z_{n}$ in ${\mathbb C}[z_{i+1}, \ldots, z_{n}]$.
\end{proposition}

We can now make a first important observation regarding the structure of elimination ideals.

\begin{proposition}\label{lm:sqsg}
Given 
 $\Vector{h}\subseteq {\mathbb Q}[z_1, \ldots, z_n]$, for any $1\leq i < n$, 
 if the elimination ideal 
 $$\langle \Vector{h}\rangle\cap {\mathbb Q}[z_1, \ldots,  z_i, z_{i+1}]=\langle q \rangle \; \text{with}\; \deg(q, z_{i+1})>0,$$
 and
 if $\sqrt{\langle q \rangle}=\langle g \rangle$, 
 then 
 \begin{enumerate}
\item  there exists an affine variety $V\subsetneq {\mathbb C}^{i}$ such that for any 
$\Vector{b}\in {\mathbb C}^{i}\backslash V$,
\begin{align}\label{eq:lsq}
\langle \Vector{h}(\Vector{b})\rangle\cap {\mathbb C}[z_{i+1}]~=~\langle q(\Vector{b})\rangle, \;\; \text{and}
\end{align}
\item
there exists an affine variety $W\subsetneq {\mathbb C}^{i}$ such that for any 
$\Vector{b}\in {\mathbb C}^{i}\backslash W$,
\begin{align}\label{eq:lsg}
\sqrt{\langle \Vector{h}(\Vector{b}) \rangle\cap {\mathbb C}[z_{i+1}]}~=~\langle g(\Vector{b})\rangle. \;\;
\end{align}
\end{enumerate}
\end{proposition}
\begin{proof}
Let $\mathcal{G}$ be a
Gr\"obner basis of $ \langle \Vector{h}\rangle$
with respect to the lexicographic order 
$z_{1}<\cdots <z_{n}$. For any $1\leq i < n$, 
let $\mathcal{N}$ be a noncomparable set of $\mathcal{G}$ with respect to $z_{i+1}, \ldots, z_n$.

Part (1): 
 If $\langle \Vector{h}\rangle\cap {\mathbb Q}[z_1, \ldots,  z_{i+1}]=\langle q \rangle$, then
by Proposition \ref{pro:elim}, ${\mathcal G}\cap \Q[z_1, \ldots, z_{i+1}]$ is a 
Gr\"obner basis of $\langle q \rangle$. So ${\mathcal G}\cap \Q[z_1, \ldots, z_{i+1}]$ contains only one element, say $h$, and hence
$h=c\cdot q$ where $c\in \Q$.
Also, $\struc{{\mathcal G}_i}={\mathcal G}\cap \Q[z_1, \ldots, z_i]=\emptyset$ since $\deg(h, z_{i+1})=\deg(q, z_{i+1})>0$, and hence, 
$\cV\left({\mathcal G}_i\right)=\C^i$.
By Proposition \ref{pro:noncomp}, 
there exists $V\subsetneq {\mathbb C}^{i}$ such that for any 
$\Vector{b}\in {\mathbb C}^{i}\backslash V$, 
$\mathcal{N}(\Vector{b})$ is a 
Gr\"obner basis of 
$\langle \Vector{h}(\Vector{b})\rangle$. 
By Proposition \ref{pro:elim}, $\mathcal{N}(\Vector{b})\cap \C[z_{i+1}]$ is a 
Gr\"obner basis of 
$\langle \Vector{h}(\Vector{b})\rangle \cap \C[z_{i+1}]$. 
Notice that $\mathcal{N}(\Vector{b})\cap \C[z_{i+1}]=\{h(\Vector{b})\}$. 
So we have 
\[\mathcal{N}(\Vector{b})\cap \C[z_{i+1}]=\langle h(\Vector{b})\rangle =\langle q(\Vector{b})\rangle. \]

Part (2): If $\sqrt{\langle q \rangle}=\langle g \rangle$, then $\cV(q)=\cV(g)$. So for any 
$\Vector{b}\in {\mathbb C}^{i}$, 
$\cV(q(\Vector{b}))=\cV(g(\Vector{b}))$.
And hence, 
$\sqrt{\langle q(\Vector{b})\rangle}=\sqrt{\langle g(\Vector{b})\rangle}$.
Note $\langle g \rangle$ is a radical ideal.
Then it is a basic fact that there exists an affine variety
$V_1\subsetneq {\mathbb C}^{i}$ such that for any 
$\Vector{b}\in {\mathbb C}^{i}\backslash V_1$, $\langle g(\Vector{b})\rangle$ is still radical, and hence $\sqrt{\langle q(\Vector{b})\rangle}=\sqrt{\langle g(\Vector{b})\rangle}=\langle g(\Vector{b})\rangle$.
By part (1), there exists an affine variety
$V_2\subsetneq {\mathbb C}^{i}$ such that for any 
$\Vector{b}\in {\mathbb C}^{i}\backslash V_2$, 
we have the equality \eqref{eq:lsq}. Let $W=V_1\cup V_2$. Then for any 
$\Vector{b}\in {\mathbb C}^{i}\backslash W$, we have 
 the equality \eqref{eq:lsg}. 
\end{proof}

\begin{proposition}\label{lm:sfactor}
Let $g\in  {\mathbb Q}[z_1, \ldots, z_n]$. If $g=\Pi_{k=1}^rg_k^{m_k}$, where every $g_k$ is irreducible in ${\mathbb Q}[z_1, \ldots, z_n]$, and $g_j\neq g_k$ for any $j\neq k$, 
then for any  $1\leq i< n$, there exists an infinite subset $\Gamma \subseteq \Q^i$ such that for any $\Vector{b}\in \Gamma$, we have $g(\Vector{b})=\Pi_{k=1}^rg_k(\Vector{b})^{m_k}$, where $g_k(\Vector{b})$ is irreducible in ${\mathbb Q}[z_{i+1}, \ldots, z_n]$, and $g_j(\Vector{b})\neq g_k(\Vector{b})$ for any $j\neq k$. 
\end{proposition}


\begin{proof}
Given that $g_1,\ldots,g_r$ are irreducible, by Hilbert's irreducibility theorem, see e.g., \cite[Theorem 1]{VGR2018}, for any $1\leq i< n$, there exists an infinite subset $\Theta \subseteq \Q^i$ such that for any $\Vector{b}\in W$,  $g_1(\Vector{b}),\ldots,g_r(\Vector{b})$ are irreducible in ${\mathbb Q}[z_{i+1}, \ldots, z_n]$. Now consider an arbitrary pair $g_j,g_k$ with $j \neq k$ and thus $g_j \neq g_k$. Without loss of generality,  let
\begin{align*}
	\struc{W_{j,k}} \ = \ \left\{ \Vector{b} \in \C^{i} \ | \ g_j(\Vector{b})- g_k(\Vector{b})=0\right\}. 
\end{align*}
Obviously, $W_{j,k}$ is an affine variety, which does not equal $\C^{i}$. Then let
\[\Gamma \ = \  \Theta \cup \cup_{k = 1}^r \cup_{j = 1}^{k-1}  \left(\Q^{i}\backslash W_{j,k}\right),\] and we are done.
\end{proof}

\subsection{General Zero-dimensional Systems}\label{sec:gzds}
In practice, a system of Lagrange likelihood equations is usually a general zero-dimensional system; see Definition \ref{def:gzd}. For instance, all models in Appendix \ref{appendix} have general zero-dimensional systems of Lagrange likelihood equations. 
So, throughout the rest of the paper, we always assume that a system of Lagrange likelihood equations is general zero-dimensional.
A general zero-dimensional system has a nice structure, see Proposition \ref{shape}, which leads us to analyze their elimination ideals furhter.

In what follows, we consider a polynomial ring ${\mathbb Q}[a_1, \ldots a_k, y_1, \dots, y_m]$ with two kinds of indeterminates: parameters $a_1, \ldots, a_k$ and variables $y_1, \ldots, y_m$.  
Also note a polynomial set $\Vector{h}=\{h_1,\ldots, h_m\}$ we are interested in has the same number of polynomials with the number of variables. 

\begin{definition}[\bf General Zero-Dimensional System]\label{def:gzd}
A finite polynomial set \[\Vector{h}=\{h_1,\ldots, h_m\}
	~\subseteq~ {\mathbb Q}[a_1, \ldots a_k, y_1, \dots, y_m]\]
is called a \struc{\textit{general zero-dimensional system}} if there exists an affine variety $V\subsetneq {\mathbb C}^k$ such that for any 
$\Vector{b}=(b_1, \ldots, b_k)\in {\mathbb C}^k\backslash V$, 
the equations $h_1(\Vector{b})=\cdots=h_m(\Vector{b})=0$  satisfy:
%
\begin{enumerate}
\item the number of complex solutions is a constant $N>0$;

\item all complex solutions are distinct;

\item every pair of distinct complex solutions $\Vector{y}^*=(y^*_1, \ldots, y^*_m)$ and $\Vector{z}^*=(z^*_1, \ldots, z^*_m)$ have distinct first coordinates, i.e.,  $y^*_1\neq z^*_1$.

\end{enumerate}
Also, we define the number $N$ stated in (1) as $\struc{N(\Vector{h})}$, i.e., 
$\struc{N(\Vector{h})}=\#{\mathcal V}(\Vector{h}(\Vector{b}))$,  where $\Vector{b}\in {\mathbb C}^k\backslash V$.
\end{definition}

\begin{proposition}
\citep[Theorem 6.10]{DMST2018}\label{shape}
Given a general zero-dimensional system
 \[\Vector{h}\subset {\mathbb Q}[a_1, \ldots, a_k, y_1, \ldots, y_m],\] 
if $\mathcal{G}$ is a
Gr\"obner basis of $ \langle \Vector{h}\rangle$
with respect to the lexicographic order 
$a_1<\cdots<a_k<y_{1}<\cdots <y_{m}$,
then there exist $m$ polynomials $T_1, \ldots, T_m\in \mathcal{G}$ such that:
\begin{enumerate}
\item  $T_1, T_2, \ldots, T_m$ have the following form: 
\begin{equation}\label{eq:triangular}
\begin{aligned}
         T_{1}~&=~C_{N}~y_1^N +C_{N-1}~y_1^{N-1} + \ldots +C_{1}~y_1 + C_{0}~, \\
	 T_{2}~ &= ~ L_{2}~ y_{2} + R_{2}~,\\
	               &~~\vdots\\
	  T_{m}~ &= ~ L_{m}~ y_{m} + R_{m}~,
\end{aligned}
\end{equation}
where 
$N~=~N(\Vector{h})$, and for each 
$i = 0,\dots, N$, $j =2,\dots, m$, we have
\begin{align*}
C_{i} \in {\mathbb Q}[a_1, \ldots, a_k], \quad 
L_{j} \in {\mathbb Q}[a_1,  \ldots, a_k], \quad 
R_{j}\in {\mathbb Q}[a_1, \ldots, a_k, y_{1}, \ldots, y_{j-1}]. 
\end{align*}
\item Let $\Vector{b}\in {\mathbb C}^k$. If $C_{N}L_2L_3\cdots L_{m}(\Vector{b})\neq 0$, then
the set $\{T_1(\Vector{b}), \ldots, T_m(\Vector{b}) \}$ is a Gr\"obner basis of	$\langle \Vector{h}(\Vector{b}) \rangle$
with respect to the lexicographic ordering $y_{1}<\cdots<y_{m}$ in ${\mathbb C}[y_1,  \ldots,  y_m]$.
\end{enumerate}
\end{proposition}


Here, we explain the relation between the Shape Lemma \cite{BMMT1994} and Proposition \ref{shape}. The original Shape Lemma \cite{BMMT1994} describes a radical zero-dimensional ideal generated by polynomials containing no parameters.
Later, a version for the system involving parameters is described by geometric resolutions \cite{GHMMP1998, GLS2000}. 
A self-contained proof of Proposition \ref{shape} in the first author's other work \cite[Theorem 6.10]{DMST2018} shows that a triangular system \eqref{eq:triangular} representing the solution set  can be selected from a Gr\"obner basis of a given general zero-dimensional ideal.  As a consequence we obtain the following theorem. 


\begin{proposition}\label{cry:generalei}
Consider a general zero-dimensional system
 \[\Vector{h}\subset {\mathbb Q}[a_1, \ldots, a_k, y_1, \ldots, y_m].\]
 If the elimination ideal $\langle \Vector{h}\rangle\cap {\mathbb Q}[a_1, \ldots, a_k, y_1]$ is principal, 
then its radical ideal is  generated by a polynomial $g\in {\mathbb Q}[a_1, \ldots, a_k, y_1]$ such that 
$\deg(g, y_1)~=~ N(\Vector{h})$.

\end{proposition}

\begin{proof}
Let $\mathcal{G}$ be a
Gr\"obner basis of $ \langle \Vector{h}\rangle$
with respect to the lexicographic order 
$a_1<\cdots<a_k<y_{1}<\cdots <y_{m}$. 
Since $\Vector{h}$ is general zero-dimensional, assume $T_1$ is the polynomial in $\mathcal{G}$ stated in Proposition \ref{shape} , which is the first polynomial in the triangular system  \eqref{eq:triangular}. 
By Proposition \ref{pro:elim}, ${\mathcal G}\cap {\mathbb Q}[a_1, \ldots, a_k, y_1]$ is a Gr\"obner basis of $\langle \Vector{h}\rangle\cap {\mathbb Q}[a_1, \ldots, a_k, y_1]$. By the hypothesis that $\langle \Vector{h}\rangle\cap {\mathbb Q}[a_1, \ldots, a_k, y_1]$ is principal,  
${\mathcal G}\cap {\mathbb Q}[a_1, \ldots, a_k, y_1]$ contains only one element.  
So we know $\{T_1\} = {\mathcal G}\cap {\mathbb Q}[a_1, \ldots, a_k, y_1]$ since 
$T_1\in {\mathcal G}\cap {\mathbb Q}[a_1, \ldots, a_k, y_1]$, and hence, $\langle \Vector{h}\rangle\cap {\mathbb Q}[a_1, \ldots, a_k, y_1] = \langle T_1 \rangle$.  By \cite[page 187, Proposition 12]{CLO2015}, 
$\sqrt{\langle T_1 \rangle}=\langle g \rangle$, where
\[g=\frac{T_1}{\struc{\gcd}(T_1, \frac{\partial T_1}{\partial a_1}, \ldots,  \frac{\partial T_1}{\partial a_k},  \frac{\partial T_1}{\partial y_1})}.\]
Here, ``$\gcd$" means the greatest common divisor. 
Hence, $\deg(g, y_1)\leq \deg(T_1, y_1)=N(\Vector{h})$. Below, we prove $\deg(g, y_1)\geq N(\Vector{h})$. 

By Definition \ref{def:gzd}, there exists an affine variety $V_1\subsetneq {\mathbb C}^k$ such that for any $\Vector{b}\in \C^{k}\backslash V_1$, 
${\mathcal V}(\Vector{h}(\Vector{b}))$
has $N(\Vector{h})$ distinct complex points with distinct $y_1$-coordinates. 
By Proposition \ref{lm:sqsg} (2), 
 there exists $V_2\subsetneq {\mathbb C}^k$ such that for any $\Vector{b}\in \C^{k}\backslash V_2$,
  \begin{align*}
\sqrt{\langle \Vector{h}(\Vector{b})\rangle\cap {\mathbb Q}[y_1]}~=~\langle g(\Vector{b})\rangle. \;\;
\end{align*}
Let $\Vector{b}\in \C^{k}\backslash (V_1\cup V_2)$. Then 
$g(\Vector{b})=0$ has $N(\Vector{h})$ distinct complex solutions, which are the $y_1$-coordinates of points in ${\mathcal V}(\Vector{h}(\Vector{b}))$. 
So $\deg(g, y_1)\geq \deg(g(\Vector{b}), y_1) \geq N(\Vector{h})$. 
\end{proof}

\subsection{Radical Elimination Ideals of Lagrange likelihood equations}

Now we apply the results of the previous sections to the special case of Lagrange likelihood equations.

 \begin{lemma}\label{lm:Nf}
 Let $\cM$ be an algebraic statistical model. If its system of Lagrange likelihood equations $\Vector{f}=\{f_0, \ldots, f_{n+s+1}\}$ \eqref{eq:lle} is  general zero-dimensional, then $N(\Vector{f})$ is equal to the ML-degree of $\cM$. 
 \end{lemma}
 
 \begin{proof}
 Let the ML-degree be $N$. By Theorem \ref{th:mld},  
 there exists an affine variety $V_1\subsetneq {\mathbb C}^{n+1}$ such that for any $\Vector{b}\in \C^{n+1}\backslash V_1$, 
$\#{\mathcal V}(\Vector{f}(\Vector{b}))=N$. 
By Definition \ref{def:gzd}, 
there exists $V_2\subsetneq {\mathbb C}^{n+1}$ such that for any $\Vector{b}\in \C^{n+1}\backslash V_2$, 
$\#{\mathcal V}(\Vector{f}(\Vector{b}))=N(\Vector{f})$. 
Pick $\Vector{b}\in \C^{n+1}\backslash (V_1\cup V_2)$. Then we have 
$N(\Vector{f})=\#{\mathcal V}(\Vector{f}(\Vector{b}))=N$.
 \end{proof}

\begin{corollary}\label{cry:ei}
Let $\cM$ be an algebraic statistical model with ML-degree $N$. Assume that its system of Lagrange likelihood equations \eqref{eq:lle}
\[\Vector{f}=\{f_0, \ldots, f_{n+s+1}\}\subseteq \Q[u_0, \ldots, u_n, p_0, \ldots, p_n, \lambda_1, \ldots, \lambda_{s+1}]=\Q[\Vector{u}, \Vector{p}, \Vector{\lambda}]\] 
 is general zero-dimensional. 
If the elimination ideal $\langle \Vector{f}\rangle\cap {\mathbb Q}[\Vector{u}, p_0]$ is principal,  
then $\sqrt{\langle \Vector{f}\rangle\cap {\mathbb Q}[\Vector{u}, p_0]}$
 is generated by a polynomial in the form 
\begin{align}\label{eq:h}
g(\Vector{u}, p_0)~=~
A_{N}(\Vector{u})~p_0^N + \cdots +A_{1}(\Vector{u})~p_0 + A_{0}(\Vector{u})
\; \;\left(A_N(\Vector{u})\neq 0,\;A_i(\Vector{u})\in  {\mathbb Q}[\Vector{u}]\right),
\end{align}
and 
\begin{enumerate}
\item there exists an affine variety $V\subsetneq {\mathbb C}^{n+1}$ such that for any 
$\Vector{b}\in {\mathbb C}^{n+1}\backslash V$,  
\begin{align}\label{eq:sf1}
\sqrt{\langle \Vector{f}(\Vector{b})\rangle\cap {\mathbb Q}[p_0]}~=~\langle g(\Vector{b}) \rangle;
\end{align}
\item for each $k=0, \ldots, n$, 
there exists an affine variety $V\subsetneq {\mathbb C}^{n}$ such that for any 
$\Vector{b}\in {\mathbb C}^{n}\backslash V$,  
\begin{align}\label{eq:sf2}
\sqrt{\langle f_0|_{\Vector{u}^{(k)}=\Vector{b}}, \ldots, f_{n+s+1}|_{\Vector{u}^{(k)}=\Vector{b}}\rangle\cap {\mathbb Q}[u_k, p_0]}~=~\langle g|_{\Vector{u}^{(k)}=\Vector{b}}\rangle, \;\;
\end{align}
where $\struc{\Vector{u}^{(k)}}=(u_1, \ldots, u_{k-1}, u_{k+1}, \ldots, u_n)$.
\end{enumerate}
\end{corollary}

\begin{proof}
By Lemma \ref{lm:Nf} and Proposition \ref{cry:generalei}, 
the radical of $\langle \Vector{f}\rangle\cap {\mathbb Q}[\Vector{u}, p_0]$ has the form \eqref{eq:h}.
Applying Proposition \ref{lm:sqsg} (2), we can conclude (1) and (2). More specifically, 
let $\Vector{h}=\Vector{f}$.
If we let $(z_1, \ldots, z_{i})=\Vector{u}$ and let $(z_{i+1}, \ldots, z_{2n+s+3})=(\Vector{p}, \Vector{\lambda})$, then $\Vector{h}\subseteq \Q[z_1, \ldots, z_{2n+s+3}]$, and by  Proposition \ref{lm:sqsg} (2), we have \eqref{eq:sf1}.
For each $k=0, \ldots, n$, if we let $(z_1, \ldots, z_{i})=(u_1, \ldots, u_{k-1}, u_{k+1}, \ldots, u_n)$ and 
let $(z_{i+1}, \ldots, z_{2n+s+3})=(u_k, \Vector{p}, \Vector{\lambda})$, 
then still, we have 
$\Vector{h}\subseteq \Q[z_1, \ldots, z_{2n+s+3}]$
, and  by  Proposition \ref{lm:sqsg} (2), we have \eqref{eq:sf2}. 
\end{proof}

\section{A Structure Theorem for Statistical Models}\label{sec:structure}

In this section we prove our theoretical main result Theorem \ref{th:main} (Corollary  \ref{cry:main}) with respect to elimination ideals of Lagrange likelihood equations,  and relate the main result to the discriminants of elimination ideals, see Corollary \ref{cry:discr}.

\subsection{Elimination Ideals of Likelihood Equations}

In this subsection we show Theorem \ref{th:main}, our main theoretical result. It can be summarized as follows: Let $\cM$ be a statistical model. If
\begin{enumerate}
	\item its defining Lagrange likelihood equations form a general zero-dimensional system, and
	\item both elimination ideals of the likelihood equations and its \struc{\textit scaled equations} are principal,
\end{enumerate}
then the sum of data appears (with a certain pattern) in the coefficients of generator polynomials of the two above elimination ideals. 
We remark that
the conditions (1--2) are satisfied for various statistical models.
We give a precise statement of the main result in Section \ref{sec:maintheorem}, and provide various examples in Section \ref{sssec:ex} for understanding the main result. 
After that, we relate the main result on elimination ideals to their discriminants in Section \ref{sec:maindiscr}. In Section \ref{sec:alg}, we will further show that  the knowledge of the structure of elimination ideals vastly simplifies the computation of the elimination ideal.

\subsubsection{Main Theorem}\label{sec:maintheorem}

In order to state the main theorem precisely, we fix some notions for the rest of this section. In what follows, $\struc{\cM}$ denotes a \struc{\textit{statistical model}} with \struc{\textit{ML-degree} $N$} and defining the system of \struc{\textit{Lagrange likelihood equations} $\Vector{f}=\{f_0,\ldots,f_{n+s+1}\}$} as in \eqref{eq:lle}. We define the polynomial
\begin{align*}
    \struc{\su} \ = \ \sum_{k=0}^n u_k \in  \Q[\Vector{u}].
\end{align*} 

Moreover, we need a scaled version of the usual Lagrange likelihood equations.

\begin{definition}[Scaled Likelihood Equations]
\label{def:ScaledLikelihoodEquations}
Let $\cM$ be an algebraic statistical model with independent model invariants $g_1, \ldots, g_s\in \Q[p_0, \ldots, p_n]$. 
Consider its Lagrange likelihood equation system  $\Vector{f}\subseteq \Q[\Vector{u},\Vector{p},\Vector{\lambda}]$ \eqref{eq:lle}. We introduce the \struc{\textit{scaled variables}} $\struc{x_j} = p_j \cdot \su$ for $j = 0,\ldots,n$, and define the system of \struc{\textit{scaled likelihood equations} $\Vector{F}=\{F_0, \ldots, F_{n+s+1}\} \subseteq \Q[\Vector{u},\Vector{x},\Vector{\lambda}]$} corresponding to $\Vector{f}$ as the numerators of the rational functions
\begin{align*}
f_{0 \ {\left|p_j=\frac{x_j}{\su}\right.}}, \ \ldots, \  f_{n+s+1 \ \left|p_j=\frac{x_j}{\su}\right.}, 
\end{align*}
i.e., when setting $\struc{d} = \max(\deg(g_1), \ldots, \deg(g_s))$, we obtain the system
\begin{align}\label{eq:slle}
\begin{array}{rl}
\struc{F_{0}(\Vector{u},\Vector{x},\Vector{\lambda})} \ = \ &\su^{d-1}x_0\lambda_1+\su^{d-\deg(g_1)}\frac{\partial g_1}{\partial p_0}(x_0, \ldots, x_n)\lambda_2 + \cdots + \\
& \su^{d-\deg(g_s)}\frac{\partial g_s}{\partial p_0}(x_0, \ldots, x_n)\lambda_{s+1}-\su^{d}u_0,\\
&\quad\vdots\\
\struc{F_{n}(\Vector{u},\Vector{x},\Vector{\lambda})} \ = \ &\su^{d-1}x_n\lambda_1+\su^{d-\deg(g_1)}\frac{\partial g_1}{\partial p_n}(x_0, \ldots, x_n)\lambda_2 + \cdots + \\
& \su^{d-\deg(g_s)}\frac{\partial g_s}{\partial p_n}(x_0, \ldots, x_n)\lambda_{s+1}-\su^{d}u_n,\\
\struc{F_{n+1}(\Vector{u},\Vector{x},\Vector{\lambda})}=&g_1(x_0, \ldots,x_n),\\
&\quad\vdots\\
\struc{F_{n+s}(\Vector{u},\Vector{x},\Vector{\lambda})} =&g_s(x_0,\ldots,x_n),\\
\struc{F_{n+s+1}(\Vector{u},\Vector{x},\Vector{\lambda})}=&x_0+\cdots+x_n-\su.
\end{array}
\end{align}
\end{definition}

For our main theorem we need to consider the generators of two specific radical ideals. 
First, 
consider Lagrange likelihood equation system $\Vector{f}$  \eqref{eq:lle}, and assume that the ideal $\langle \Vector{f} \rangle \cap \Q[\Vector{u},p_0]$ is principal. 
	Then we denote the generator of its radical ideal by $\struc{\h}$, i.e.
	\begin{align*}
        \struc{\langle \h \rangle} \ = \ \sqrt{\langle  \Vector{f} \rangle \cap \Q[\Vector{u},p_0]}.
	\end{align*}
Second, for scaled likelihood equation system $\Vector{F}$ \eqref{eq:slle}, also assume $\langle  \Vector{F}  \rangle \cap \Q[\Vector{x},p_0]$ is principal.  We define 
	\begin{align*}
        \struc{\langle \F \rangle} \ = \ \sqrt{\langle \Vector{F}  \rangle \cap \Q[\Vector{u},x_0]}.
	\end{align*}
Furthermore, we define
\begin{align*}
	\struc{\ell_{\Vector{f}}} \ = \
	\begin{cases}
		1 & \text{ if } \h \in \langle \su \rangle \\
		0 & \text{ if } \h \notin \langle \su \rangle
	\end{cases}, \qquad
	\struc{\delta_{\Vector{f}}} \ = \
	\begin{cases}
		1 & \text{ if } \F \in \langle \su \rangle \\
		0 & \text{ if } \F \notin \langle \su \rangle
	\end{cases}.
\end{align*}
When $\Vector{f}$ and $\Vector{F}$ are clear from the context, we  simply denote $\ell_{\Vector{f}}$ and  $\delta_{\Vector{F}}$ by $\struc{\ell}$ and $\struc{\delta}$.
	
\begin{theorem}[Main Theorem]\label{th:main}
Let $\cM$ be an algebraic statistical model with ML-degree $N$, Lagrange likelihood equation system $\Vector{f}$ in \eqref{eq:lle}, and scaled likelihood equation system $\Vector{F}$ as in \eqref{eq:slle}. 
Assume that 
\begin{enumerate}
	\item $\Vector{f}$ is general zero-dimensional, and
	\item both ideals $\langle \Vector{f}\rangle\cap {\mathbb Q}[\Vector{u}, p_0]$ and $\langle \Vector{F}\rangle\cap {\mathbb Q}[\Vector{u}, x_0]$ are principal.
\end{enumerate}
Then there exist an integer $t>0$, and a polynomial $C(\Vector{u})\in  {\mathbb Q}[\Vector{u}]\setminus \langle \su \rangle$  such that 
\begin{align}\label{eq:MainTheoremStructureEF}
\F|_{x_0=\su \cdot p_0}~=~C(\Vector{u})\su^{t-\ell}\h,
\end{align}
and, $\h$ has the form
\begin{align}\label{eq:MainTheoremStructureEf}
\h \ = \ \sum_{k=0}^NB_{k}(\Vector{u})\su^{k-t+\ell}~p_0^k, \;\;\; \text{where}\;\; B_k(\Vector{u})=\frac{\coeff(\F, x_0^k)}{C(\Vector{u})}\in {\mathbb Q}[\Vector{u}].
\end{align}
\end{theorem}

While the previous theorem is our main theoretical result, it is not easy to see its relevance. We clarify this in the following corollary.

\begin{corollary}\label{cry:main}
Let all notions and assumptions be as in Theorem \ref{th:main}. 
Then there exist an integer $t >0$, and a polynomial $C(\Vector{u})\in  {\mathbb Q}[\Vector{u}]\setminus \langle \su \rangle$  such that for $k=0, \ldots, N$, we have 
\begin{align}\label{eq:crymain}
 \coeff(\h, p_0^k)~\su^{t-\ell-\delta} \ = \  \tilde{B}_k(\Vector{u})~\su^{k}~, \;\;\; \text{where}\;\; \tilde{B}_k(\Vector{u})=\frac{\coeff(\F, x_0^k)}{\su^{\delta}C(\Vector{u})} \in {\mathbb Q}[\Vector{u}].
\end{align}
Therefore, 
\begin{enumerate}
\item   if $k>t - \ell-\delta$, 
then the sum of data $\su$ appears in $\coeff(\h, p_0^k)$ as a factor with
multiplicity at least $k - t + \ell+\delta$, 
\item if $k<t - \ell-\delta$, then $\su$ appears in $\coeff(\F, x_0^k)$ as a factor with
multiplicity at least $t -\ell - k$,  and 
\item only when $k = t - \ell-\delta$, it is possible for either $\coeff(\h, p_0^k)$ or $\coeff(\F, x_0^k)$ contains no factor $\su$. 
\end{enumerate}
\end{corollary}
\begin{proof}
	Rewrite \cref{eq:MainTheoremStructureEf} as 
	$\sum_{k=0}^N\coeff(\h, p_0^k)~p_0^k  =  \sum_{k=0}^NB_{k}(\Vector{u})\su^{k-t+\ell}~p_0^k$, 
where $B_k(\Vector{u})=\frac{\coeff(\F, x_0^k)}{C(\Vector{u})}$. Comparing the coefficients from both sides, we have   \[\coeff(\h, p_0^k) =B_{k}(\Vector{u})\su^{k-t+\ell}  =\frac{\coeff(\F, x_0^k)}{\su^{\delta}C(\Vector{u})}\su^{k-t+\ell+\delta},\] 
	and so we have \eqref{eq:crymain}.
\end{proof}

We summarize the values of $N, t, \ell, \delta$ stated in Corollary \ref{cry:main} for more test models in Table \ref{tab:Nt}. We explain how to (efficiently) compute $N, t, \ell, \delta$ by Example \ref{ex:t}. 

\begin{table}[h]
\centering
\begin{tabular}{ | c | c | c | c | c | c | c | c | c | c |  }
\hline
Models  & \ref{ex:l1} & \ref{ex:l2} & \ref{ex:l3} & \ref{ex:l4} & \ref{ex:l5} &  \ref{ex:l6} & \textcolor{red}{\ref{ex:l7}} &  \ref{ex:l8} & \ref{ex:l9} \\ 
 \hline
 $N$ & 3 & 2 & 4 & 6 & 12 & 10 & \textcolor{red}{23} &14 & 9 \\  
 \hline
$t$ & 2 & 1 & 2 & 2 & 2 & 1 & \textcolor{red}{20} & 1 & 3 \\
 \hline   
 $\ell$ & 0 & 0 & 0 & 0 & 0 & 0 & 0 & 0 & 0 \\
 \hline
 $\delta$ & 1 & 1 & 1 & 1 & 1 & 1 & 1 & 1 & 1 \\
 \hline
\end{tabular}
\smallskip 
\caption{$N, t, \ell, \delta$ for the models tested in this article; see Section \ref{sec:implementation}. The defining equations of Models \ref{ex:l1}--\ref{ex:l9} are listed in the Appendix \ref{appendix}. }\label{tab:Nt}
\end{table}
From some examples and computational results for concrete models (see Table \ref{tab:Nt} and Section \ref{sssec:ex}),  we remark that:
\begin{enumerate}
\item in practice, $t$ is indeed usually a small number (Table \ref{tab:Nt});
\item  we have $\ell = 0$ in all investigated examples (Table \ref{tab:Nt}), but we do not have a proof for this fact;
\item both cases $\delta=1$ (Example \ref{ex:linear}) and $\delta=0$ (Example \ref{ex:coin})  can happen. 
\end{enumerate}
By Corollary  \ref{cry:main} (1), for a given model with ML-degree $N$,  if $t-\ell-\delta$ is much less than $N$, then
$\su$ appears in most coefficients of $\h$: $\coeff(\h, p_0^{t-\ell-\delta+1}), \ldots, \coeff(\h, p_0^N)$.  This fact will be used to improve the efficiency for computing 
$\h$ in Section \ref{sec:alg}.

\subsubsection{Examples for Main Theorem}\label{sssec:ex}
\begin{example}[Four-Sided Die]\label{ex:linear}
Consider the linear model $\cM$ below given by a weighted four-sided die \cite[Example 1]{Tang2017}, for which we know the ML-degree is $3$. 
\[\cM~=~{\mathcal P}\left(p_0+2p_1+3p_2-4p_3\right)\cap \Delta_3,\]
where 
\[\Delta_3~=~\{(p_0, p_1, p_2, p_3)\in {\mathbb R}^4_{>0} |p_{0}+ p_{1}+p_{2}+ p_{3}=1\}.\]
The Lagrange likelihood equations \eqref{eq:lle} are 
  \begin{align*}
f_0&=p_{0}\lambda_1+p_{0}\lambda_2-  u_{0}& f_3&=p_{3}\lambda_1-4p_{3}\lambda_2-   u_{3}\\
f_1&=p_{1}\lambda_1+2p_{1}\lambda_2 -  u_{1}& f_4&=p_0+2p_1+3p_2-4p_3\\
f_2&=p_{2}\lambda_1+3p_{2}\lambda_2-  u_{2}& f_5&=p_{0} + p_{1} +p_{2}+p_{3} - 1,
\end{align*}
where $p_{0},p_{1},p_{2},p_{3}, \lambda_1, \lambda_2$ are variables, and  $u_{0},  u_{1},  u_{2},  u_{3}$ are parameters. 
Let $\su=u_0+u_1+u_2+u_3$. 
The scaled equations \eqref{eq:slle} are 
  \begin{align*}
F_0&=x_{0}\lambda_1+x_{0}\lambda_2-  \su u_{0}& F_3&=x_{3}\lambda_1-4x_{3}\lambda_2- \su  u_{3}\\
F_1&=x_{1}\lambda_1+2x_{1}\lambda_2 - \su  u_{1}& F_4&=x_0+2x_1+3x_2-4x_3\\
F_2&=x_{2}\lambda_1+3x_{2}\lambda_2- \su u_{2}& F_5&=x_{0} + x_{1} +x_{2}+x_{3} - \su.
\end{align*}
By computing a Gr\"obner basis, we can verify that $\sqrt{\langle f_0, \ldots, f_{5}\rangle\cap {\mathbb Q}[u_0, u_1, u_2, u_3, p_0]}$ is generated by 
{\footnotesize
\begin{align}\label{eq:ex1h}
\h~=~10\su^2p_0^3 -(43u_0+20u_1+15u_2+8u_3)\su p_0^2 + 2u_0(29u_0+23u_1+21u_2+14u_3)p_0 -24u_0^2,
\end{align}
}
and we also verify that  $\sqrt{\langle F_0, \ldots, F_{5}\rangle\cap {\mathbb Q}[u_0, u_1, u_2, u_3, x_0]}$ is generated by 
{\footnotesize
\begin{align}\label{eq:ex1F}
\F~=~\su\left(10x_0^3 -(43u_0+20u_1+15u_2+8u_3)x_0^2 + 2u_0(29u_0+23u_1+21u_2+14u_3)x_0 -24u_0^2\su\right).
\end{align}}Both elimination ideals are principal. So the hypotheses of Theorem \ref{th:main} are satisfied. Below, we show \eqref{eq:MainTheoremStructureEF} and \eqref{eq:MainTheoremStructureEf} in Theorem \ref{th:main} hold.  
\begin{enumerate} 
\item Substituting $x_0=\su p_0$ into $\F$, we have  $\F|_{x_0=\su p_0}$ is equal to 
{\footnotesize
\begin{align}\label{eq:ex1sF}
\su^2\left(10\su^2p_0^3 -(43u_0+20u_1+15u_2+8u_3)\su p_0^2 + 2u_0(29u_0+23u_1+21u_2+14u_3) p_0 -24u_0^2\right).
\end{align}}Comparing \eqref{eq:ex1h} and \eqref{eq:ex1sF}, we see $\F|_{x_0=\su p_0} = \su^2 \h$. 
So we have \eqref{eq:MainTheoremStructureEF}, where $t=2$, $\ell=0$ and $C(\Vector{u})=1$. 
\item Following (1), $t=2$. By \eqref{eq:ex1h}, $\h$ can be written as the form \eqref{eq:MainTheoremStructureEf}:
\begin{align}\label{eq:ex1h2}
\h~=~B_3\su p_0^3 +B_2 p_0^2 + B_1\su^{-1}p_0 +B_0\su^{-2}
\end{align}
where 
\begin{equation*}
\begin{tabular}{ll}
$B_3=10\su$,&$B_2=(43u_0+20u_1+15u_2+8u_3)\su$,\\
$B_1=2u_0(29u_0+23u_1+21u_2+14u_3)\su$,&$B_0=-24u_0^2\su^{2}$.
\end{tabular}
\end{equation*}
\end{enumerate}
Now we consider Corollary \ref{cry:main}. From \eqref{eq:ex1F}, $\F\in  \langle \su\rangle$, so we have $\delta=1$.  Also, from this example we see $\F\in  \langle \su\rangle$  means $\su$ will be a common factor of those above $B_k$'s.
Let $\tilde{B}_k=\frac{B_k}{\su}$.
Then it is seen from \eqref{eq:ex1h2} that $\h$ can be also written as:
\begin{align*}
\h~=~\tilde{B}_3\su^2 p_0^3 +\tilde{B}_2\su p_0^2 + \tilde{B}_1p_0 +\tilde{B}_0\su^{-1}.
\end{align*}
Now it is straightforward to see the equality \eqref{eq:crymain} in Corollary \ref{cry:main}. 
More than that, from $\eqref{eq:ex1h}$ and $\eqref{eq:ex1F}$, we see that 
if $k>t-\ell-\delta=1$, then $\su$ appears in $\coeff(\h, p_0^k)$, and
if $k<t-\ell-\delta=1$, then $\su$ appears in $\coeff(\F, p_0^k)$. 
Also for $k=t-\ell-\delta=1$, we see that $\coeff(\h, p_0^k)$ contains no factor $\su$. 
\end{example}




\begin{example}\label{ex:t}
In this example, we explain by the linear model in Example \ref{ex:linear} on how to compute $N, t, \ell, \delta$ presented  in  Table \ref{tab:Nt} when we can not easily get $\h$ and $\F$. Consider Lagrange likelihood equations $f_0, \ldots, f_5$ and scaled equations $F_0, \ldots, F_5$ in Example \ref{ex:linear}.
For each $u_j\neq u_0$, substitute $u_j= b_j$ into $f_0, \ldots, f_5, F_0, \ldots, F_5$, where $b_j$ is a random  rational number. For instance, 
we substitute $u_1=2, u_2=12, u_3=7$ and rename the resulting polynomials as $f^*_0, \ldots, f^*_5, F^*_0, \ldots, F^*_5$.  
By computing a Gr\"obner basis, we first get a generator of $\sqrt{\langle f^*_0, \ldots, f^*_{5}\rangle\cap {\mathbb Q}[u_0, p_0]}$:
\begin{align}\label{eq:tg1}
{\footnotesize
g^*(u_0, p_0)~=~10(u_0+21)^2p_0^3-(u_0+21)(43u_0+276)p_0^2+2u_0(29u_0+396)p_0-24u_0^2.
}
\end{align}
Similarly, we compute $\sqrt{\langle F^*_0, \ldots, F^*_{5}\rangle\cap {\mathbb Q}[u_0, x_0]}$, and get
\[G^*(u_0, x_0)~=~\left(u_0+21\right)\left(10x_0^3-(43u_0+276)p_0^2+2u_0(29u_0+396)p_0-24(u_0+21)u_0^2\right).\]
\begin{enumerate}
\item By Proposition \ref{lm:sqsg} (2), we have $g^*=\h|_{u_1=2, u_2=12, u_3=7}$, and $G^*=\F|_{u_1=2, u_2=12, u_3=7}$.
By Corollary \ref{cry:ei}, ML-degree $N$ is $\deg(\h, p_0)=\deg(g^*, p_0)=3$. 
\item Notice that $\su|_{u_1=2, u_2=12, u_3=7}=u_0+21$. So
 \[\left(\F|_{x_0=\su p_0}\right)|_{u_1=2, u_2=12, u_3=7}=G^*|_{x_0=(u_0+21)p_0}.\]
 Substituting $x_0=(u_0+21)p_0$ into $G^*$, we have  $G^*|_{x_0=(u_0+21)p_0}$ is
\begin{align}\label{eq:tg2}
{\scriptsize
(u_0+21)^2\left(10(u_0+21)^2p_0^3-(u_0+21)(43u_0+276)p_0^2+2u_0(29u_0+396)p_0-24u_0^2\right).
}
\end{align}
Comparing \eqref{eq:tg1} and \eqref{eq:tg2},  we have $G^*|_{x_0=(u_0+21)p_0}=(u_0+21)^2g^*(u_0, p_0)$. So
 \[\left(\F|_{x_0=\su p_0}\right)|_{u_1=2, u_2=12, u_3=7}~=~\left(\su^2\h
\right)|_{u_1=2, u_2=12, u_3=7}.\] 
By Proposition \ref{lm:sfactor}, $\F|_{x_0=\su p_0}=\su^2\h$, and hence, the integer $t$ stated in Theorem \ref{th:main} is $2$. 
\item  Notice $u_0+21$ is not a factor  of $g^*$. By Proposition \ref{lm:sfactor}, $\su$ is not a factor of $\h$. So $\h\not \in \langle \su\rangle$, and hence, $\ell=0$.
\item  Notice $u_0+21$ is a factor  of $G^*$. By Proposition \ref{lm:sfactor}, $\su$ is a factor of $\F$. So $\F \in \langle \su\rangle$,  and hence, $\delta=1$.  
\end{enumerate}
\end{example}


\begin{example}[Fair Coin]\label{ex:coin}
In Table \ref{tab:Nt}, we have $\delta=1$ for all test models. Here, we show one example with $\delta=0$. 
Consider the model given by a fair $2$-sided coin
\[\cM~=~{\mathcal P}\left(p_0-p_1\right)\cap \Delta_1, \text{ where }
\Delta_1~=~\{(p_0, p_1)\in {\mathbb R}^2_{>0} |p_{0}+ p_{1}=1\}.\]
The scaled equations \eqref{eq:slle} are 
  \begin{align*}
F_0&=x_{0}\lambda_1+\lambda_2- (u_0 + u_1)  u_{0}&F_2&=x_0-x_1\\
F_1&=x_{1}\lambda_1-\lambda_2 - (u_0 + u_1) u_{1}& F_3&=x_{0} + x_{1} - (u_0 + u_1).
\end{align*}
By computing a Gr\"obner basis, we can verify that $\sqrt{\langle F_0, \ldots, F_{3}\rangle\cap {\mathbb Q}[u_0,  u_1, x_0]}$ is generated by $2x_0-u_0-u_1$, which does not belong to the ideal $\langle \su\rangle$. 
\end{example}

\subsubsection{Proof of the Main Theorem}
In this subsection we prove Theorem \ref{th:main}. Before the proof, we prepare some notions, definitions and lemmata.

\begin{definition}[Scaling Map]\label{def:scalingmap}
We define the \struc{\textit{scaling map} $\phi$} as:
\begin{align*}
	\struc{\phi}: {\mathbb C}^{n+1}\times {\mathbb C}^{n+1}\times {\mathbb C}^{s+1} ~\to~ {\mathbb C}^{n+1}\times {\mathbb C}^{n+1}\times {\mathbb C}^{s+1}, \quad (\Vector{u}, \Vector{p}, \Vector{\lambda}) \mapsto (\Vector{u}, {\mathcal S}(\Vector{u}) \cdot \Vector{p}, \Vector{\lambda}),
\end{align*}
where $\struc{{\mathcal S}(\Vector{u}) \cdot \Vector{p}}$ denotes the vector $({\mathcal S}(\Vector{u}) \cdot p_0, \ldots, {\mathcal S}(\Vector{u}) \cdot p_n)$. Similarly, we define the \struc{\textit{truncated scaling map} $\phi_0$} as:
\begin{align*}
	\struc{\phi_0}: {\mathbb C}^{n+1}\times {\mathbb C} ~\to~ {\mathbb C}^{n+1}\times {\mathbb C}, \quad (\Vector{u}, p_0) \mapsto (\Vector{u}, {\mathcal S}(\Vector{u}) \cdot p_0).
\end{align*}
\end{definition}

Below, we denote  by $\struc{\cV(\cS)}$ the affine variety generated by $\su=\sum_{k=0}^nu_k$ in $\C^{n+1}$, which is a hyperplane. In what follows we will also have to consider $\su$ as a polynomial embedded in the rings $\Q[\Vector{u}, p_0]$ and $\Q[\Vector{u},\Vector{p},\Vector{\lambda}]$. We define the hyperplanes corresponding to these embedings as
\begin{align*}
	\struc{\SmallAmbientVarietySU} = \cV(\cS) \times \C \ \text{ and } \ \struc{\BigAmbientVarietySU} = \cV(\cS) \times \C^{n+1} \times \C^{s+1}.
\end{align*}

\begin{remark}	\label{rem:ScalingMapIsomorphism}
The scaling map $\phi$ is isomorphic on 
	\begin{align*}
		\left({\mathbb C}^{n+1}\times {\mathbb C}^{n+1}\times {\mathbb C}^{s+1}\right) \setminus{\BigAmbientVarietySU} ~\to~ \left({\mathbb C}^{n+1}\times {\mathbb C}^{n+1}\times {\mathbb C}^{s+1}\right) \setminus{\BigAmbientVarietySU} 
	\end{align*}
		Analogously, the truncated scaling map is isomorphic on
	\begin{align*}
		\left({\mathbb C}^{n+1}\times \mathbb C\right) \setminus{\SmallAmbientVarietySU} ~\to~ \left({\mathbb C}^{n+1}\times \mathbb C\right) \setminus{\SmallAmbientVarietySU}.
	\end{align*}
\end{remark}

\begin{lemma}\label{lm:s0}
Given a system of scaled likelihood equations $\Vector{F} = \{F_0, \ldots, F_{n+s+1}\}$ \eqref{eq:slle}, 
we have 
\[\cV(\cS) \times \{0\}^{n+1}\times {\mathbb C}^{s+1}\subseteq \cV\left(\Vector{F}\right)\cap\BigAmbientVarietySU.\]
\end{lemma}

\begin{proof}
Clearly, $\cV(\cS) \times \{0\}^{n+1}\times {\mathbb C}^{s+1}\subseteq \BigAmbientVarietySU$. 
We only need to prove 
for any $\Vector{u}\in \cV(\cS)$ and for any $\Vector{\lambda}\in \C^{s+1}$,  
$(\Vector{u}, \Vector{0}, \Vector{\lambda})\in \cW\left(\Vector{F}\right)$, where $\Vector{0} \in {\mathbb C}^{n+1}$.
In fact, by \eqref{eq:slle}, since $\Vector{\lambda}$ does not appear in $F_{n+s+1}$ and ${\mathcal S}(\Vector{u})=0$, we have $F_{n+s+1}$ vanishes at $(\Vector{u}, \Vector{0}, \Vector{\lambda})$. 
By  \eqref{eq:slle} and Definition \ref{definition:StatisticalModel}, $g_j(\Vector{x})$ $(j=1, \ldots, s)$ in $F_{n+1}, \ldots, F_{n+s}$ are homogenous. So 
$F_{n+1}, \ldots, F_{n+s}$ vanish at any point with all zero-$x_i$ coordinates. Also, because ${\mathcal S}(\Vector{u})=0$, and because the partial derivatives $\frac{g_j}{p_k}(\Vector{x})$
in $F_0, \ldots, F_{n}$ are still homogenous, so $F_0, \ldots, F_n$ vanish at $(\Vector{u}, \Vector{0}, \Vector{\lambda})$.
\end{proof}

\begin{lemma}\label{lm:inverse}
Given a system of scaled likelihood equations $\Vector{F} $ \eqref{eq:slle}, 
\begin{align}\label{eq:inverse}
\phi^{-1}\left(\cV\left(\Vector{F}\right)\cap\BigAmbientVarietySU\right) \;=\; \BigAmbientVarietySU, \;\; \text{and}\; \;
\phi_0^{-1}\left(\proj_{n+2}\left(\cV\left(\Vector{F}\right)\cap\SmallAmbientVarietySU\right)\right) \;=\; \SmallAmbientVarietySU.
\end{align}
\end{lemma}

\begin{proof}
We only prove the first equality since the argument for the second one is similar. By Definition \ref{def:scalingmap}, the map $\phi$ does not change the $\Vector{u}$-coordinate, so
$\phi^{-1}(\BigAmbientVarietySU) \subseteq \BigAmbientVarietySU$, and hence, 
$\phi^{-1}\left(\cV\left(\Vector{F}\right)\cap\BigAmbientVarietySU\right) \subseteq \BigAmbientVarietySU$.
On the other hand, by 
Lemma \ref{lm:s0},  $\cV(\cS) \times \{0\}^{n+1}\times {\mathbb C}^{s+1}\subseteq \cV\left(\Vector{F}\right)\cap\BigAmbientVarietySU$. 
So $\phi(\BigAmbientVarietySU)=\cV(\cS) \times \{0\}^{n+1}\times {\mathbb C}^{s+1} \subseteq \cV\left(\Vector{F}\right)\cap\BigAmbientVarietySU$, and hence, $\BigAmbientVarietySU \subseteq \phi^{-1}\left(\cV\left(\Vector{F}\right)\cap\BigAmbientVarietySU\right)$.
\end{proof}

\begin{lemma}\label{lm:VFf}

Given a system of Lagrange likelihood equations $\Vector{f} $ \eqref{eq:lle} and its corresponding system of scaled equations
$\Vector{F}$ \eqref{eq:slle}, we have 
\begin{align}\label{eq:VFf}
\phi^{-1}\left(\cW\left(\Vector{F}\right)\right)~=~\cV\left(\Vector{f}\right) \cup \BigAmbientVarietySU.
\end{align}
In particular, $\cW\left(\Vector{F}\right)$ and $\cV(\Vector{f})$ are birational.
\end{lemma}

\begin{proof}
It follows by the Definitions \ref{definition:LikelihoodEquations} and \ref{def:ScaledLikelihoodEquations}, and Remark \ref{rem:ScalingMapIsomorphism} that
\begin{align}
	\phi^{-1}\left(\cW\left(\Vector{F} \right) \setminus \BigAmbientVarietySU\right) \ = \ \cV(\Vector{f} ) \setminus \BigAmbientVarietySU, \;\text{and} \;
	\phi\left(\cV(\Vector{f} ) \setminus \BigAmbientVarietySU\right) \ = \ \cW\left(\Vector{F} \right) \setminus \BigAmbientVarietySU.
	\label{eq:BirationalityProofIsomorphism}
\end{align}
The birationality follows with \eqref{eq:BirationalityProofIsomorphism}. Moreover, by Lemma \ref{lm:inverse}, we have $\phi^{-1}\left(\cV\left(\Vector{F}\right)\cap\BigAmbientVarietySU\right) = \BigAmbientVarietySU$, and so \eqref{eq:VFf} is proved. 
\end{proof}

\begin{lemma}\label{cry:piphiF}
Given a system of Lagrange likelihood equations $\Vector{f}$ \eqref{eq:lle} and its corresponding system of scaled equations
$\Vector{F}$ \eqref{eq:slle},
we have 
\begin{align}\label{eq:piphiF}
\proj_{n+2}\left(\cV\left(\Vector{f}\right) \cup \BigAmbientVarietySU\right)~=~\phi_0^{-1}\left(\proj_{n+2}\left(\cW\left(\Vector{F}\right)\right)\right)
\end{align}
\end{lemma}

\begin{proof}
By Lemma \ref{lm:VFf}, it is sufficient to show 
\begin{align}\label{eq:lmpiphiF}
\proj_{n+2}\left(\phi^{-1}\left(\cW\left(\Vector{F}\right)\right)\right)~=~\phi_0^{-1}\left(\proj_{n+2}\left(\cW\left(\Vector{F}\right)\right)\right).
\end{align}
Since $\phi$ and $\phi_0$ are isomorphic everywhere except on $\BigAmbientVarietySU$ and $\SmallAmbientVarietySU$ respectively, 
the following diagram commutes when $(\Vector{u}, \Vector{p}, \Vector{\lambda})\not\in \BigAmbientVarietySU$:
\begin{center}	
\begin{tikzpicture}[node distance=2.5cm]
  \node (C) {$(\Vector{u}, \Vector{p}, \Vector{\lambda})$};
  \node (Bi) [right of=C] {$(\Vector{u}, \su\Vector{p}, \Vector{\lambda})$};
  \node (P) [below of=C] {$(\Vector{u}, p_0)$};
  \node (Ai) [right of=P] {$(\Vector{u}, \su p_0)$};
  \draw[->] (C) to node [above]{~~$\phi$~~} (Bi);
  \draw[->] (C) to node [left] {$\proj_{n+2}$~~~} (P);
  \draw[->] (P) to node [above] {~~$\phi_0$~~} (Ai);
   \draw[->] (Bi) to node [right] {$\proj_{n+2}$~~~} (Ai);
\end{tikzpicture}	
\end{center}
So, it is sufficient to show the equality \eqref{eq:lmpiphiF} for 
$\cW\left(\Vector{F}\right) \cap \BigAmbientVarietySU$. By Lemma \ref{lm:inverse}, we have
\begin{align*}
	\proj_{n+2}\left(\phi^{-1}\left(\cW\left(\Vector{F}\right) \cap \BigAmbientVarietySU \right) \right)\ = \
	\proj_{n+2}\left(\BigAmbientVarietySU\right) \ = \ 
	\SmallAmbientVarietySU \ = \ \phi_0^{-1}\left(\proj_{n+2}\left(\cV\left(\Vector{F}\right)\cap\SmallAmbientVarietySU\right)\right)
	\end{align*}
\end{proof}

\begin{lemma}\label{lm:dimF}
If a system of Lagrange likelihood equations $\Vector{f}$ \eqref{eq:lle} is general zero-dimensional with ML-degree $N$, then 
the corresponding system of scaled equations 
$\Vector{F}$ \eqref{eq:slle} is also  general zero-dimensional, and 
$N(\Vector{F})=N$.
\end{lemma}

\begin{proof}
We have to prove the three properties of a general zero-dimensional system in Definition \ref{def:gzd}.
We consider the birational map $\phi$. By Lemma \ref{lm:VFf} we have 
\begin{align*}
	\phi^{-1}\left(\cW\left(\Vector{F}\right)\right)~=~\cV\left(\Vector{f}\right) \cup \BigAmbientVarietySU.  
\end{align*}
Hence, for any  $\Vector{b}\in \C^{n+1}\backslash \BigAmbientVarietySU$, we have
\begin{align*}
	\phi^{-1}\left(\{\Vector{b}\}\times \cW\left(\Vector{F}(\Vector{b})\right)\right)~=~\{\Vector{b}\}\times \cV\left(\Vector{f}(\Vector{b})\right).
\end{align*}
By Remark \ref{rem:ScalingMapIsomorphism}, $\phi$ is an isomorphism on the restricted set $\C^{2n+s+3}\backslash \BigAmbientVarietySU$.
Thus, in particular
\begin{align*}
	\# \cW\left(\Vector{F}(\Vector{b})\right) \ = \ \# \cV\left(\Vector{f}(\Vector{b})\right) \ = \ N.
\end{align*}
So $\Vector{F}$ satisfies condition (1) of Definition \ref{def:gzd}. 
Since $\Vector{f}$ is general zero-dimensional, the first entries of every pair of distinct points in $\cV\left(\Vector{f}(\Vector{b})\right) $ are distinct, i.e., we have $\# \proj_1\left(\cV\left(\Vector{f}(\Vector{b})\right)\right) = N$. 
That means $\# \proj_{n+2}\left(\{\Vector{b}\}\times \cV\left(\Vector{f}(\Vector{b})\right)\right) = N$.
Since $\phi$ restricted to $\Vector{u}$ is the identical map, we conclude
\begin{align*}
	\# \proj_{n+2}(\{\Vector{b}\}\times\cV(\Vector{F}(\Vector{b})) \ = \ \# \proj_{n+2}(\phi(\{\Vector{b}\}\times\cV(\Vector{f}(\Vector{b})))) \ = \ N.
\end{align*}
This implies moreover that all complex solutions of $\cV(\Vector{F}(\Vector{b}))$ are distinct with distinct first entries since all complex solutions of $\cV(\Vector{f}(\Vector{b}))$ are distinct with distinct first entries, which implies that $\Vector{F}$ satisfies conditions (2) and (3) of Definition \ref{def:gzd}.

\end{proof}

\begin{corollary}\label{cry:eiF}
Given an algebraic statistical model $\cM$, assume its  Lagrange likelihood equation system $\Vector{f}$ \eqref{eq:lle} is general zero-dimensional with ML-degree $N$.  Let  the scaled equation system  be $\Vector{F}$ \eqref{eq:slle}. 
If the elimination ideal $\langle \Vector{F}\rangle\cap {\mathbb Q}[\Vector{u}, x_0]$ is principal, then the radical of this elimination ideal is generated by a polynomial in the form 
\begin{align}\label{eq:hF}
H_{N}(\Vector{u})~x_0^N +H_{N-1}(\Vector{u})~x_0^{N-1} + \ldots +H_{1}(\Vector{u})~x_0 + H_{0}(\Vector{u}),
\end{align}
where  $H_i(\Vector{u})\in  {\mathbb Q}[\Vector{u}]$ and $H_N(\Vector{u})\neq 0$. 
\end{corollary}
\begin{proof}
The conclusion follows from Proposition \ref{cry:generalei} and Lemma \ref{lm:dimF}. 
\end{proof}

\begin{lemma}\label{lm:key}
Given a system of Lagrange likelihood equations $\Vector{f}$ \eqref{eq:lle} and the system of  scaled likelihood equations $\Vector{F}$ \eqref{eq:slle}, if
both ideals $\langle \Vector{f}\rangle\cap {\mathbb Q}[\Vector{u}, p_0]$ and $\langle \Vector{F}\rangle\cap {\mathbb Q}[\Vector{u}, x_0]$ are principal, then we have 
\begin{align}\label{eq:key}
\pV\left(\su\cdot\h\right)~\subset~\pV\left(\F|_{x_0=\su \cdot p_0}\right)
\end{align}
\end{lemma}
\begin{proof}
Recall $\langle \h \rangle=\langle \Vector{f} \rangle\cap \Q[\Vector{u}, p_0]$. By Proposition \ref{pro:clos}, 
$\overline{\proj_{n+2}\left(\cV\left(\Vector{f}\right)\right)}~=~\pV\left(\h\right)$.
So we first have
\begin{align}\label{eq:projectf}
\pV\left(\su\cdot \h\right)
~=~\overline{\proj_{n+2}\left(\cV\left(\Vector{f}\right)\right)}\cup\overline{\proj_{n+2}(\BigAmbientVarietySU}) 
~=~\overline{\proj_{n+2}(\cV\left(\Vector{f}\right) \cup \BigAmbientVarietySU)}.
\end{align}
We next prove
\begin{align}\label{eq:projectF}
\overline{\phi_0^{-1}\left(\proj_{n+2}\left(\cW\left(\Vector{F}\right)\right)\right)}~\subset~\pV\left(\F|_{x_0=\su \cdot p_0}\right)
\end{align}
Since $\pV\left(\F|_{x_0=\su \cdot p_0}\right)$ is closed, we only need to show 
$\phi_0^{-1}\left(\proj_{n+2}\left(\cW\left(\Vector{F}\right)\right)\right)~\subset~\pV\left(\F|_{x_0=\su \cdot p_0}\right)$.
In fact, recall $\langle \F \rangle=\langle \Vector{F} \rangle\cap \Q[\Vector{u}, x_0]$. By  Proposition \ref{pro:clos}, 
$\overline{\proj_{n+2}\left(\cW\left(\Vector{F}\right)\right)}=\xV\left(\F\right)$.
So, for any $(\Vector{u}, p_0)\in \phi_0^{-1}\left(\proj_{n+2}\left(\cW\left(\Vector{F}\right)\right)\right)$, 
\[\phi_0\left(\Vector{u}, p_0\right)=\left(\Vector{u}, {\mathcal S}(\Vector{u})p_0
\right)\in \proj_{n+2}\left(\cW\left(\Vector{F}\right)\right)\subseteq \xV\left(\F\right).\] 
 Hence, $(\Vector{u}, p_0)\in \pV\left(\F|_{x_0=\su \cdot p_0}\right)$. Now the equality \eqref{eq:key} follows from \eqref{eq:projectf}, Lemma \ref{cry:piphiF},  and \eqref{eq:projectF}:
$$\pV\left(\su\cdot \h\right)~=~\overline{\proj_{n+2}(\cV\left(\Vector{f}\right) \cup \BigAmbientVarietySU)}~=~\overline{\phi_0^{-1}\left(\proj_{n+2}\left(\cW\left(\Vector{F}\right)\right)\right)}~\subset~\pV\left(\F|_{x_0=\su \cdot p_0}\right).$$
\end{proof}

Now, we can finally prove Theorem \ref{th:main}.

\noindent

\begin{proof}[Proof of Theorem \ref{th:main}.]~
First, we prove the equality \eqref{eq:MainTheoremStructureEF}. Recall that we define $\ell = 1$ if $\h\in \langle \su\rangle$ and $\ell = 0$ otherwise.
Since $\langle \h\rangle$ is radical, $\h$ is squarefree.
Then 
 $\sqrt{\langle \su\cdot\h \rangle} = \langle  \su^{1-\ell}\cdot \h \rangle$. By Lemma \ref{lm:key},  we have 
$\F|_{x_0=\su \cdot p_0}\in \sqrt{\langle \su\cdot\h \rangle}$. Hence, there exists a $\hat C\in {\mathbb Q}[\Vector{u}, p_0]$ such that 
\begin{equation}\label{eq:mainproof1}
\F|_{x_0=\su \cdot p_0}=\hat C\cdot \su^{1 - \ell} \cdot \h.
\end{equation} 
Next, we show 
$\hat C\in {\mathbb Q}[\Vector{u}]$. Assume we had $\hat C\in {\mathbb Q}[\Vector{u}, p_0]\backslash {\mathbb Q}[\Vector{u}]$. 
Then 
$\deg\left(\hat C, p_0
\right)>0$.
By Corollary \ref{cry:ei} it holds that $\deg\left(\su^{1-\ell} \h, p_0\right)=\deg\left(\h, p_0\right)=N$. 
Thus, by \eqref{eq:mainproof1}, 
\[\deg\left(\F, x_0 \right) \geq 
\deg\left(\F|_{x_0=\su \cdot p_0}, p_0
\right)=\deg\left(\hat C, p_0
\right)+\deg\left(\su^{1-\ell} \h, p_0
\right)>N.\]
However, by Corollary \ref{cry:eiF},  $\deg\left(\F, x_0
\right)=N$. By contradiction it follows that $\hat C\in {\mathbb Q}[\Vector{u}]$. 
Let $t$ be the smallest positive integer such that $\hat C\in \langle \su^{t-1}\rangle$ and $\hat C\not\in \langle \su^{t}\rangle$.
Then we can write 
\[\hat C=\su^{t-1} C(\Vector{u}),\; \text{where}\; C(\Vector{u})\not\in \langle \su\rangle.\] 
So, by \eqref{eq:mainproof1},  
\begin{equation}\label{eq:mainproof2}
\F|_{x_0=\su \cdot p_0} \ = \  \hat C\cdot \su^{1 - \ell} \cdot \h \ = \ C(\Vector{u})\su^{(t-1)+(1-\ell)} \h \ = \ C(\Vector{u})\su^{t- \ell} \h.
\end{equation} 
Therefore, the equality \eqref{eq:MainTheoremStructureEF} is proven. 

Second, we prove \eqref{eq:MainTheoremStructureEf}. By Corollary \ref{cry:eiF}, assume $\F~=~ \sum_{k=0}^N H_{k}(\Vector{u})~x_0^k$, and hence,
\begin{align*}
	  \F|_{x_0=\su \cdot p_0} ~=~ \sum_{k=0}^N H_{k}(\Vector{u})~\su^k p_0^k.
\end{align*}
So, by  \eqref{eq:mainproof2}
we have 
\[\h ~=~\sum_{k=0}^N \frac{H_{k}(\Vector{u})}{C(\Vector{u})}~\su^{k-t+\ell} p_0^k.\]
In the equation above, $\h$ is a polynomial. So $\frac{H_{k}(\Vector{u})}{C(\Vector{u})}~\su^{k-t+\ell}$ on the right side is a polynomial for every $k$.  
Note also from the previous proof, $C(\Vector{u})\not\in \langle \su \rangle$. So, $\frac{H_{k}(\Vector{u})}{C(\Vector{u})}$ is a polynomial in ${\mathbb Q}[\Vector{u}]$. 
Notice $H_{k}(\Vector{u})=\coeff(\F, x_0^k)$.
Hence, the equality  
\eqref{eq:MainTheoremStructureEf} is proven.
\end{proof}

\subsection{Discriminants of Elimination Ideals}\label{sec:maindiscr}
In this subsection we relate our results from Theorem \ref{th:main} to discriminants, which we introduce as a next step. We follow the definition of Gelfand, Kapranov, and Zelevinski; see \cite[page 405, Formula (4.30) and notion on page 411]{gkzbook}.

\begin{definition}
Consider $f(z)=\sum_{k=0}^Nc_kz^k\in {\mathbb Q}[c_0, \ldots, c_N][z]$ as a polynomial in $z$ with $\deg(f, z)=N$ and general coefficients $c_0, \ldots, c_N$.
We denote the \struc{\textit{discriminant}} of $f(z)$ with respect to $z$ by $\struc{\mathrm{discr}(f; z)}$. 
By \cite{gkzbook}, $\mathrm{discr}(f; z)$ is a homogenous polynomial in ${\mathbb Q}[c_0, \ldots, c_N]$:
\begin{align*}
	\ \sum_{\varphi_0,\ldots,\varphi_N}C_{\varphi_0,\ldots,\varphi_N}c_0^{\varphi_0}\cdots c _N^{\varphi_N}.
\end{align*}
Note in the above formula, the exponent vectors $(\varphi_0, \ldots, \varphi_N)$ and coefficients $C_{\varphi_0,\ldots,\varphi_N}$ only depend on $N$. So 
we denote it by $\struc{\mathcal{D}_N(c_0, \ldots, c_N)}$.
\label{def:discriminant}	
\end{definition}


In practice, we observed that the polynomial $\su$ regularly appears as a factor in the discriminant of $\h(\Vector{u}, p_0)$ with respect to $p_0$ (Example \ref{ex:discr}). In the following Corollary \ref{cry:discr} we explain this observation:  it says if the integer $t$ stated in Theorem \ref{th:main} is less than $\frac{N}{2}$, then $\su$ must be a factor of the discriminant. For most  models in our experiments, we do have $t<\frac{N}{2}$; see Table \ref{tab:Nt}.

\begin{corollary}\label{cry:discr}
Let all notions and assumptions be as in Theorem \ref{th:main}. 
 Then there exists  an integer $t>0$ and $C(\Vector{u})\in  {\mathbb Q}[\Vector{u}]\setminus \langle \su \rangle$ such that 
\begin{align}\label{eq:d}
\mathrm{discr}(\h; p_0)~=~\su^{(N-2(t-\ell-\delta))(N-1)}~{\mathcal D}_N(\tilde{B}_0, \ldots, \tilde{B}_N), \;\;\; \text{where}\;\; \tilde{B}_k=\frac{\coeff(\F, x_0^k)}{\su^{\delta} C(\Vector{u})}\in {\mathbb Q}[\Vector{u}]. 
\end{align}
\end{corollary}
\begin{proof}
Assume
\begin{align*}
\h~=~A_{N}(\Vector{u})~p_0^N +A_{N-1}(\Vector{u})~p_0^{N-1} + \ldots +A_{1}(\Vector{u})~p_0 + A_{0}(\Vector{u}).
\end{align*}
Then, we can further write
\[\mathrm{discr}(\h, p_0)~=~{\mathcal D}_N({A}_0, \ldots, {A}_N)~=~\sum_{\varphi_0,\ldots,\varphi_N}C_{\varphi_0,\ldots,\varphi_N}A_0^{\varphi_0}\cdots A _N^{\varphi_N}.\] 
By Corollary \ref{cry:main}, there exists an integer $t>0$ such that 
for every $k=0, \ldots, N$, $A_k(\Vector{u})=\tilde{B}_k(\Vector{u})\su^{k-t+\ell+\delta}$, where $\tilde{B}_k(\Vector{u})=\frac{\coeff(\F, x_0^k)}{\su^{\delta} C(\Vector{u})}\in {\mathbb Q}[\Vector{u}]$.  Then 
\begin{align}\label{eq:dish}
\mathrm{discr}(\h; p_0)~=~\sum_{\varphi_0,\ldots,\varphi_N}C_{\varphi_0,\ldots,\varphi_N}\su^{\sum^N_{k=0}(k-t+\ell+\delta)\varphi_k}\tilde{B}_0^{\varphi_0}\cdots \tilde{B}_N^{\varphi_N}.
\end{align}
By \cite[Theorem 2.2, page 412]{gkzbook}, 
\begin{align}\label{eq:gkz}
\sum^{N}_{k=0}\varphi_k~=~2N-2, \;\;\text{and}\;\; \sum^{N}_{k=0}(N-k)\varphi_k~=~N(N-1).
\end{align}
So, we have 
\begin{align}\label{eq:gkz1}
 \sum^{N}_{k=0}k\varphi_k \ = \ N\sum^{N}_{k=0}\varphi_k- \sum^{N}_{k=0}(N-k)\varphi_k=N(2N-2)-N(N-1)=N(N-1).
 \end{align}
 By \eqref{eq:gkz} and \eqref{eq:gkz1}, $\sum^N_{k=0}(k-t+\ell+\delta)\varphi_k$ is equal to  
\begin{align}\label{eq:gkz2}
 \sum^N_{k=0}k\varphi_k-(t-\ell-\delta)\sum^N_{k=0}\varphi_k=N(N-1)-(t-\ell-\delta)(2N-2)=(N-2(t-\ell-\delta))(N-1).
\end{align}
Hence, by  \eqref{eq:dish} and \eqref{eq:gkz2}, 
\[\mathrm{discr}(\h; p_0)~=~\su^{(N-2(t-\ell-\delta))(N-1)}\sum_{\varphi_0,\ldots,\varphi_N}C_{\varphi_0,\ldots,\varphi_N}\tilde{B}_0^{\varphi_0}\cdots \tilde{B}_N^{\varphi_N}
.
\]
By Definition \ref{def:discriminant}, we have $\sum_{\varphi_0,\ldots,\varphi_N}C_{\varphi_0,\ldots,\varphi_N}\tilde{B}_0^{\varphi_0}\cdots \tilde{B}_N^{\varphi_N} ={\mathcal D}_N(\tilde{B}_0, \ldots, \tilde{B}_N)$.  So, the equality \eqref{eq:d} is proven.

\end{proof}

\begin{example}[Example \ref{ex:linear} Continued]\label{ex:discr}
By Example \ref{ex:linear}
 and Corollary \ref{cry:discr}, 
the discriminant of $\h$ in \eqref{eq:ex1h} has the factor $\su^{(N-2(t-\ell-\delta))(N-1)}|_{N=3, t=2, \ell=0, \delta=1}=(u_0+u_1+u_2+u_3)^2$. 
Moreover, by Corollary \ref{cry:discr} (Equality \eqref{eq:d}), we have
\begin{align*}
\mathrm{discr}(\h; p_0)~=~(u_0+u_1+u_2+u_3)^{2}~{\mathcal D}_3(\tilde{B}_0, \tilde{B}_1, \tilde{B}_2, \tilde{B}_3), 
\end{align*}
where by Example \ref{ex:discr}, 
\begin{align}\label{eq:exBt}
\begin{tabular}{ll}
$\tilde{B}_3=10$,&$\tilde{B}_2=(43u_0+20u_1+15u_2+8u_3)$,\\
$\tilde{B}_1=2u_0(29u_0+23u_1+21u_2+14u_3)$,&$\tilde{B}_0=-24u_0^2(u_0+u_1+u_2+u_3)$,
\end{tabular}
\end{align}
and ${\mathcal D}_3(c_0, \ldots, c_3)=-27c_0^2c_3^2+18c_0c_1c_2c_3-4c_0c_2^3-4c_1^3c_3+c_1^2c_2^2$, 
which can be computed by running {\tt discrim}$(\sum_{i=0}^3c_iz^i, z)$ in {\tt Maple}.
So, we get $\mathrm{discr}(\h; p_0)$ below by simply substituting $c_i=\tilde{B}_i$ into ${\mathcal D}_3$: 
\begin{center}
 $4u_0^2(u_0+u_1+u_2+u_3)^2(441u_0^4+4998u_0^3u_1+20041u_0^2u_1^2+33320u_0u_1^3+19600u_1^4-756u_0^3u_2+
20034u_0^2u_1u_2+83370u_0u_1^2u_2+79800u_1^3u_2-5346u_0^2u_2^2+55890u_0u_1u_2^2+119025u_1^2u_2^2+4860u_0u_2^3+76950u_1u_2^3+18225u_2^4-1596u_0^3u_3-11116u_0^2u_1u_3-17808u_0u_1^2u_3+4480u_1^3u_3+7452u_0^2u_2u_3-7752u_0u_1u_2u_3+49680u_1^2u_2u_3-17172u_0u_2^2u_3+71460u_1u_2^2u_3+27540u_2^3u_3+2116u_0^2u_3^2+6624u_0u_1u_3^2-4224u_1^2u_3^2-9528u_0u_2u_3^2+15264u_1u_2u_3^2+14724u_2^2u_3^2-1216u_0u_3^3-512u_1u_3^3+3264u_2u_3^3+256u_3^4)$, 
\end{center}
One can verify the above discriminant by
running ${\tt discrim}(\h, p_0)$  in {\tt Maple}, 
which will give a consistent result. 
\end{example}

\section{Algorithm}\label{sec:alg}

In this section we explain our main Algorithm \ref{interpolation} and its sub algorithms and also provide the corresponding pseudocode.

\medskip

Given an algebraic model $\cM$,  let $\Vector{f}=\{f_{0},\ldots, f_{n+s+1}\}\subseteq {\mathbb Q}[\Vector{u}, \Vector{p}, \Vector{\lambda}]$ be its Lagrange likelihood equation system. Assume the hypotheses (1--2) of Theorem \ref{th:main} are satisfied.
 In this section, we propose a probabilistic algorithm for computing the polynomial $\h$ generating
  $\sqrt{\langle \Vector{f}\rangle \cap \Q[\Vector{u},p_0]}$.
 We simply denote $\coeff(\h, p_0^i)$ 
 by $\struc{A_i(\Vector{u})}$. 
 Then $\h=\sum_{i=0}^NA_i(\Vector{u})p_0^i$, where $N$ is the ML-degree of $\cM$.
 First, we observe a fact by \eqref{eq:ex1h} in Example \ref{ex:linear} that:
\begin{description}
\item[({\bf F1})] $\h$ is homogenous with respect to $\Vector{u}$, and hence each $A_i(\Vector{u})$ is homogenous with the same total degree in $\Q[\Vector{u}]$. 
 \end{description}
 We omit the proof of ({\bf F1}) here since the argument is similar to 
\cite[Proposition 2]{Tang2017}, which is based on a basic fact implied by \eqref{eq:lle}:
for every $(\Vector{u}, p_0)\in \proj_{n+2}(\cV(\Vector{f}))$ and for any scalar $\gamma \neq 0$, $(\gamma \Vector{u}, p_0)$ is also in 
$\proj_{n+2}(\cV(\Vector{f}))$. 

Besides observing ({\bf F1}), we make the following assumptions to simplify our algorithm:
\begin{description}
\item[({\bf A1})] Assume $\deg(A_N(\Vector{u}), u_0)=\deg(A_N(\Vector{u}))$, i.e., $A_N(\Vector{u})$ contains a term $u_0^{\deg(A_N)}\in \Q[u_0]$. 
\item[({\bf A2})]Assume  $A_N(\Vector{u})$ is monic with respect to $u_0$, which unifies our output polynomial $\h$.
\end{description}
 
If ({\bf A1}) does not hold, then we apply an invertible linear change to 
the parameters $u_j$  such that ({\bf A1}) holds for the new parameters (similar to \cite[Algorithm 4]{Tang2017}). For instance, obtain new parameters $v_j$ as 
\[v_0 = u_0, \;\;\text{and} \;\; v_j = b_j u_j + u_0\;\;\text{for}\;\; j=1, \ldots, n,\]
where $b_j$ are randomly chosen rational numbers. 
 By \cite[Lemmas 1--2]{Tang2017},  $\deg(A_N(\Vector{v}), v_0)$ will be equal to  $\deg(A_N( \Vector{v}))$.
 However,  the linear change may cause computational expense to the subsequent computation, more specifically, the sampling step: Algorithm \ref{sampleLC} in Section \ref{sec:code}. 
In practice, we have ({\bf A1}) holds for all models we have computed. We conjecture it is always true, but we do not have a proof.  

According to Corollary \ref{cry:main}, we further write 
\begin{align*}
	\h(\Vector{u}, p_0)~=~\sum_{i=0}^NA_{i}(\Vector{u})~p_0^i~=~\sum_{i=0}^N\su^{\alpha_i}\B_{i}(\Vector{u})~p_0^i, \;
	\text{ where } \;R_i\in \Q[{\Vector{u}}]\backslash \langle \su\rangle.
\end{align*}
The main algorithm has three steps; see Algorithm \ref{interpolation} with three sub-algorithms in Section \ref{sec:code}: 
\begin{itemize}
\item[Step 1.] Compute $N$, $(\alpha_0, \ldots, \alpha_N)$, and the degree of every $u_j$ in each $A_i$;
\item[Step 2.] Compute the leading coefficient $A_N(\Vector{u})$ by interpolating $\B_N(\Vector{u})$;
\item[Step 3.] Compute the coefficients $A_i(\Vector{u})$  by interpolating $\B_i(\Vector{u})$ for $i=0, \ldots, N-1$.
\end{itemize}
The pseudocode is given in Section \ref{sec:code} and a running example is given in  Section \ref{sec:runex}.

\subsection{Pseudocode}\label{sec:code}~
In Algorithm \ref{interpolation} and its sub-algorithms Algorithms \ref{degree}--\ref{sample}, we only have finite loops. So the algorithm terminates for sure.   The algorithm is probabilistic. Proposition \ref{lm:sfactor} and Corollary \ref{cry:ei} guarantee that the probabilistic algorithm  terminates correctly if all random chosen rational numbers involved are generic. We explain more details by a running example in Section \ref{sec:runex}.
\begin{algorithm}[t]\label{interpolation}
\scriptsize
\DontPrintSemicolon
\LinesNumbered
\SetKwInOut{Input}{input}
\SetKwInOut{Output}{output}
\Input{Lagrange likelihood equations $f_0, \cdots, f_{n+s+1}
\in {\mathbb Q}[\Vector{u}, \Vector{p}, \Vector{\lambda}]$
}
\Output{A generator of $\sqrt{\langle f_0, \ldots, f_{n+s+1}\rangle\cap {\mathbb Q}[\Vector{u}, p_0]}$: $\h(\Vector{u}, p_0)=\sum_{i=0}^NA_{i}(\Vector{u})~p_0^i$}
\caption{({\bf Main Algorithm})}
\BlankLine
\tiny
$N, (\alpha_0, \ldots, \alpha_N), \dL, \dD\leftarrow $ {\bf Degrees}$(f_0, \ldots, f_{n+s+1})$\;  
$A_{N}(\Vector{u})\leftarrow$ {\bf LeadingCoefficient}$(f_0, \ldots, f_{n+s+1}, \alpha_N, \dL)$\;  
$A_{0}(\Vector{u}), \ldots, A_{N-1}(\Vector{u})\leftarrow$ {\bf Coefficients}$(f_0, \ldots, f_{n+s+1}, A_{N}(\Vector{u}), \dD)$\;
{\bf Return}  $\sum_{i=0}^NA_{i}(\Vector{u})~p_0^i$\;
\end{algorithm}

\begin{algorithm}[t]\label{degree}
\scriptsize
\DontPrintSemicolon
\LinesNumbered
\SetKwInOut{Input}{input}
\SetKwInOut{Output}{output}
\Input{ Lagrange likelihood equations $f_0, \cdots, f_{n+s+1}
\in {\mathbb Q}[\Vector{u}, \Vector{p}, \Vector{\lambda}]$ 
}
\Output{  $N, (\alpha_0, \ldots, \alpha_N), \dL, \dD$, where 
\begin{itemize}
\item $N~=~\deg(\h, p_0)$, 
               \item $\alpha_i$ is the multiplicity of the factor $\sum_{k=0}^nu_k$ appearing in $\coeff(\h, p_0^i)$, 
               \item $\dL$ is a list with length $n+1$, whose $(j+1)$-th entry is $\deg(\lcoeff(\h, p_0), u_j)$ 
               for $j=0, \ldots, n$,
            \item   $\dD$ is an $N\times (n+1)$ matrix, 
            whose $(i+1, j+1)$-entry is $\deg(\coeff(\h, p_0^i), u_j)$ 
             for $i=0, \ldots, N-1$ and for
            $j=0, \ldots, n$.
            \end{itemize}
            }
 \BlankLine
 \tiny
 $f_0^*, \ldots, f_{n+s+1}^*\leftarrow$ replace $u_1,\ldots, u_n$ in $f_0, \ldots, f_{n+s+1}$ with random rational numbers $b_1, \ldots, b_n$\; 
$g(u_0, p_0)\leftarrow$ generator of $\sqrt{\langle f_0^*, \ldots, f_{n+s+1}^*\rangle\cap {\mathbb Q}[u_0, p_0]}$
\nllabel{elim22}\;
$N\leftarrow \deg(g, p_0)$\;
\For{$i$ {\bf from} $0$ {\bf to} $N$}{
$\alpha_i\leftarrow$ multiplicity of the factor $u_0+\sum_{k=0}^n b_k$ in  $\coeff\left(g, p_0^i\right)$\; 
}
$\dL(1)\leftarrow \deg(\lcoeff(g, p_0), u_0)$\;
\For{$i$ {\bf from} $0$ {\bf to} $N-1$}{
$\dD(i+1, 1)\leftarrow  \deg(\coeff(g, p_0^i), u_0)$\;
}
\For {$j$ {\bf from} $1$ {\bf to} $n$} {
$f_0^*, \ldots, f_{n+s+1}^*\leftarrow$ replace $u_0, \ldots, u_{j-1}, u_{j+1}, \ldots, u_n$ in $f_0, \ldots, f_{n+s+1}$ with random rational numbers\; 
$g(u_j, p_0)\leftarrow$ generator of $\sqrt{\langle f_0^*, \ldots, f_{n+s+1}^*\rangle\cap {\mathbb Q}[u_j, p_0]}$\nllabel{elim28}\;
$\dL(j+1)\leftarrow \deg(\lcoeff(g, p_0), u_j)$\;
\For{$i$ {\bf from} $0$ {\bf to} $N-1$}{
$\dD(i+1, j+1)\leftarrow  \deg(\coeff(g, p_0^i), u_j)$\;
}
}
{\bf return} $N, (\alpha_0, \ldots, \alpha_N), \dL, \dD$\;
\caption{{\bf (Sub-Algorithm of Algorithm \ref{interpolation}) Degrees}}
\end{algorithm} 


\begin{algorithm}[t]\label{leadcoeff}
\scriptsize
\DontPrintSemicolon
\LinesNumbered
\SetKwInOut{Input}{input}
\SetKwInOut{Output}{output}
\Input{  Lagrange likelihood equations $f_0, \ldots, f_{n+s+1}$,  and $a_N, \dL$, where
\begin{itemize}
\item $\alpha_N$ is  the multiplicity of the factor $\sum_{k=0}^nu_k$ appearing in $\lcoeff(\h, p_0)$, 
\item $\dL$ is a list, whose $(j+1)$-th entry is $\deg(\lcoeff(\h, p_0), u_j)$ 
 for $j=0, \ldots, n$.
\end{itemize}
}
\Output{ 
$\lcoeff(\h, p_0)$: $A_{N}(\Vector{u})$ 
}
 \BlankLine
 \tiny
$d\leftarrow \dL(1)-\alpha_N$ \textcolor{blue}{\it \#Here, $d=\deg(R_N, u_0)$, and by ({\bf A1}), $\deg(R_N, u_0)=\deg(R_N)$}\;
 \For {$i$ {\bf from} $0$ {\bf to} $d-1$\nllabel{et1}}
{

Enumerate all the monomials in the set $\{u_1^{\beta_1}\cdots u_n^{\beta_n}|\sum_{j=1}^n\beta_j=d-i, 0\leq \beta_j\leq \dL(j+1)-\alpha_N\}$ as
 ${U}_{i, 1}, \ldots, {U}_{{i, t_i}}$ 
}

$t\leftarrow \max(t_0, \ldots, t_{d-1})$\;\nllabel{i1}
\For {$k$ {\bf from} $1$ {\bf to} $t$\nllabel{et2}}
{
$b_{ k, 1}, \ldots, b_{k, n}\leftarrow$ random rational numbers\;
$q(u_0)\leftarrow$ {\bf IntersectForLC}$(f_0, \ldots, f_{n+s+1}, b_{k, 1}, \ldots, b_{k, n}, \alpha_N)$\;
$C^*_{0, k}, \ldots, C^*_{d-1, k}\leftarrow$ the coefficients of $q(u_0)$ with respect to $u_0^0, \ldots, u_0^{d-1}$\;
}
\For {$i$ {\bf from} $0$ to $d-1$}
{
${\mathcal M}_i\leftarrow $ the $t_i\times t_i$ matrix whose $(k, r)$-entry  is  $U_{i, r}|_{u_1=b_{ k, 1}, \ldots, u_n=b_{k, n}}$\nllabel{matrix}\;
$C_i(u_1, \ldots, u_n)\leftarrow ({U}_{i, 1}, \ldots, {U}_{{i, t_i}}){\mathcal M}_i^{-1}(C^*_{i, 1}, \ldots, C^*_{i, t_i})^{T}$
\nllabel{i2}}
{\bf Return}  $\left(\sum_{k=0}^nu_k\right)^{\alpha_N}\left(u_0^{d} + \Sigma_{i=0}^{d-1}C_{i}\left(u_1, \ldots, u_n\right)u_0^i\right)$
\caption{{\bf (Sub-Algorithm of Algorithm \ref{interpolation}) LeadingCoefficient}}
\end{algorithm}

\begin{algorithm}[t]\label{sampleLC}
\tiny
\DontPrintSemicolon
\LinesNumbered
\SetKwInOut{Input}{input}
\SetKwInOut{Output}{output}
\Input{  Lagrange likelihood equations $f_0, \ldots, f_{n+s+1}$, random rational numbers $b_1, \ldots, b_n$ and $\alpha_N$, \\
where $a_N$ is  the multiplicity of the factor $\sum_{k=0}^nu_k$ appearing in $\lcoeff(\h, p_0)$.}
\Output{ 
$\frac{\lcoeff(\h, p_0)}{(\sum_{k=0}^nu_k)^{\alpha_N}}|_{u_1=b_1, \ldots, u_n=b_n}$}
\BlankLine
$f_0^*, \ldots, f_{n+s+1}^*\leftarrow$ replace $u_1, \ldots,  u_n$ in $f_0, \ldots, f_{n+s+1}$ with $b_1, \ldots, b_n$, respectively\; 
$g(u_0, p_0)\leftarrow$ generator of the radical of elimination ideal $\langle f_0^*, \ldots,f_{n+s+1}^*\rangle\cap {\mathbb Q}[u_0, p_0]$\nllabel{elim31}\;
$q(u_0)\leftarrow $ divide $\lcoeff(g, p_0)$ by $(u_0+\sum_{i=1}^nb_i)^{\alpha_N}$\;
Make $q(u_0)$ monic with respect to $u_0$, and 
{\bf return} $q(u_0)$\;
\caption{{\bf (Sub-Algorithm of Algorithm \ref{leadcoeff}) IntersectForLC}}
\end{algorithm}

\begin{algorithm}[t]\label{coeff}
\scriptsize
\DontPrintSemicolon
\LinesNumbered
\SetKwInOut{Input}{input}
\SetKwInOut{Output}{output}
\Input{  Lagrange likelihood equations  $f_0, \ldots, f_{n+s+1}$, and $A_{N}(\Vector{u}), \alpha_0, \ldots, \alpha_{N-1}, \dD$, where\\
\begin{itemize}
\item $A_N(\Vector{u})~=~\lcoeff(\h, p_0)$, 
\item $\alpha_i$ is the multiplicity of the factor $\sum_{k=0}^nu_k$ appearing in $\coeff(\h, p_0^i)$, 
\item $\dD$ is an $N\times (n+1)$ matrix, whose $(i+1, j+1)$-entry is $\deg(\coeff(\h, p_0^i), u_j)$. 
\end{itemize}
}
\Output{
$\coeff(\h, p_0^i)$ for $i=0, \ldots, N-1$: $A_{0}(\Vector{u}), \ldots, A_{N-1}(\Vector{u})$ 
}
\BlankLine
$d\leftarrow \deg(A_N(\Vector{u}))$\;
 \For {$i$ {\bf from} $0$ {\bf to} $N-1$\nllabel{et1}}
{
 Enumerate all the monomials in
$\{u_0^{\beta_0}\cdots u_n^{\beta_n}|\sum_{j=0}^n\beta_j=d-\alpha_i, 0\leq \beta_j\leq D(i+1, j+1)-\alpha_i\}$
as ${U}_{i, 1}, \ldots, {U}_{{i, t_i}}$
}

$t\leftarrow \max(t_0, \ldots, t_{N-1})$\;\nllabel{i1}
\For {$k$ {\bf from} $1$ {\bf to} $t$\nllabel{et2}}
{
$b_{ k, 0}, \ldots, b_{k, n}\leftarrow$ random rational numbers\;
$g(p_0)\leftarrow $ {\bf Intersect}$(f_0, \ldots, f_{n+s+1}, b_{k, 0}, \ldots, b_{k, n})$\;
$C^*_{0, k}, \ldots, C^*_{N-1, k}\leftarrow$ the coefficients of $g(p_0)$ with respect to $p_0^0, \ldots, p_0^{N-1}$\;
}
\For {$i$ {\bf from} $0$ to $N-1$}
{
${\mathcal M}_i\leftarrow $ $t_i\times t_i$ matrix whose $(k, r)$-entry  is  $\frac{U_{i, r}}{A_{N}(\Vector{u})}|_{u_0=b_{ k, 0}, \ldots, u_n=b_{k, n}}$\nllabel{matrix}\;
$\B_i(\Vector{u})\leftarrow ({U}_{i, 1}, \ldots, {U}_{{i, t_i}}){\mathcal M}_i^{-1}(C^*_{i, 1}, \ldots, C^*_{i, t_i})^{T}$
\nllabel{i2}}
{\bf Return}  $\left(\sum_{k=0}^nu_k\right)^{\alpha_0}\B_0(\Vector{u}), \ldots, \left(\sum_{k=0}^nu_k\right)^{\alpha_{N-1}}\B_{N-1}(\Vector{u})$\;
\caption{{\bf (Sub-Algorithm of Algorithm \ref{interpolation}) Coefficients}}
\end{algorithm}

\begin{algorithm}[t]\label{sample}
\scriptsize
\DontPrintSemicolon
\LinesNumbered
\SetKwInOut{Input}{input}
\SetKwInOut{Output}{output}
\Input{ Lagrange likelihood equations $f_0, \ldots, f_{n+s+1}$,  and random rational numbers $b_0, b_1, \ldots, b_n$}
\Output{ 
$\frac{\h(\Vector{u}, p_0)}{\lcoeff(\h, p_0)}|_{u_0=b_0, \ldots, u_n=b_n}$}
\BlankLine
$\tilde{f}_0, \ldots, \tilde{f}_{n+s+1}\leftarrow$ replace $u_0, \ldots,  u_n$ in $f_0, \ldots, f_{n+s+1}$ with $b_0, \ldots, b_n$, respectively\; 
$g(p_0)\leftarrow$ generator of the radical of elimination ideal $\langle \tilde{f}_0, \ldots,\tilde{f}_{n+s+1}\rangle\cap {\mathbb Q}[p_0]$\nllabel{elim61}\;
Make $g(p_0)$ monic with respect to $p_0$, and {\bf return} $g(p_0)$\;
\caption{{\bf (Sub-Algorithm of Algorithm \ref{coeff}) Intersect}}
\end{algorithm}

\subsection{Running Example}\label{sec:runex}
In this subsection, we illustrate how Algorithm \ref{interpolation} works by the four-sided-die model in Example \ref{ex:linear}. 
The inputs are $f_0, \ldots, f_5$ in Example \ref{ex:linear}, and the output will be a generator of $\sqrt{\langle f_0, \ldots, f_{5}\rangle\cap {\mathbb Q}[\Vector{u}, p_0]}$.  Assume the generator is 
\begin{align*}
	\h(\Vector{u}, p_0)~=~\sum_{i=0}^NA_{i}(\Vector{u})~p_0^i~=~\sum_{i=0}^N\su^{\alpha_i}\B_{i}(\Vector{u})~p_0^i, 
\end{align*}
where $\su=u_0+u_1+u_2+u_3$, and $R_i\in \Q[{\Vector{u}}]\backslash \langle \su\rangle$. 

\noindent
{\bf Step 1.}
First, we compute $N, (\alpha_0, \ldots, \alpha_N)$,  and $\deg(A_i, u_j)$
 for  $j=0, \ldots, 3$ and for $i=0, \ldots, N$.   
For each $u_j\neq u_0$, substitute $u_j= b_j$ into $f_0, \ldots, f_5$, where $b_j$ is a random  rational number. For instance, 
we choose $\Vector{b}=(b_1, b_2, b_3)=(2, 12, 7)$. 
We substitute $u_j=b_j$, and rename the resulting polynomials as $f^*_0, \ldots, f^*_5$.  Note that 
$f^*_k= f_k(u_0, \Vector{b}, \Vector{p}, \Vector{\lambda})$. 
We obtain a generator of $\sqrt{\langle f^*_0, \ldots, f^*_{5}\rangle\cap {\mathbb Q}[u_0, p_0]}$ by computing a Gr\"obner basis:
\[g^*(u_0, p_0)~=~10(u_0+21)^2p_0^3-(u_0+21)(43u_0+276)p_0^2+2u_0(29u_0+396)p_0-24u_0^2.\]
If $\Vector{b}$ is generic in the parameter space $\C^3$, then by Corollary \ref{cry:ei} (2), 
$g^*(u_0, p_0)=\h(u_0, \Vector{b}, p_0)$. 
So, we have \[N=\deg(\h(\Vector{u}, p_0), p_0)= \deg(\h(u_0, \Vector{b}, p_0), p_0)= \deg(g^*,p_0) = 3.\]
And, for $i=0, \ldots, N (=3)$, we  have \[\deg(A_i(\Vector{u}), u_0) = \deg(A_i(u_0, \Vector{b}), u_0)=\deg(\coeff(g^*, p_0^i), u_0)=2.\] 
So, we record $\dL(1)=\deg(A_3, u_0)=2$ and $\dD(i+1, 1)=\deg(A_i, u_0)=2$ for $i=0, 1, 2$.  Similarly, we compute the degrees of other parameters, and have 
\[\dL~=~\left[2, 2, 2, 2\right],\;\; \text{and}\;\;
\dD~=~\left[
\begin{array}{cccc}
2&0&0&0 \\
2&1&1&1 \\
2&2&2&2
\end{array}
\right],
\]
where $\dL(j+1)$ records $\deg(A_3, u_j)$,
 and 
$\dD(i+1, j+1)$ records $\deg(A_i, u_j)$ for $i=0, 1, 2$.  Notice $\mathcal{S}(u_0, \Vector{b})=u_0+21$. By Proposition \ref{lm:sfactor}, checking the multiplicity of 
the factor $u_0+21$ in each $\coeff(g^*, p_0^i)$ for $i=0, \ldots, 3$, we have $\alpha_0=\alpha_1=0$, $\alpha_2=1$, and $\alpha_3=2$. 

\bigskip
\noindent
{\bf Step 2.} The second step is to recover the leading coefficient $A_N(\Vector{u})$.
By {\bf Step 1}, we know $N=3$ and $\alpha_3=2$. We write $A_N$ as $A_3(\Vector{u})=\su^2\B_3(\Vector{u})$. By the degrees recored in $\dL$, we know the degrees of $u_0, u_1, u_2, u_3$ in $A_3(\Vector{u})$ are all $2$. 
So, $\B_3(\Vector{u})\in {\mathbb Q}$.  According to the assumption (A2), $A_3(\Vector{u})$ is monic with respect to $u_0$. Hence, $\B_3(\Vector{u})=1$, and therefore,   
$A_3(\Vector{u})=\su^2$.

\bigskip
\noindent
{\bf Step 3.} The last step is to interpolate the coefficients $A_0(\Vector{u}), A_1(\Vector{u})$ and $A_{2}(\Vector{u})$. 
As an example, we show how to interpolate $A_2(\Vector{u})$ in details.
By {\bf Step 1}, we have $\alpha_2=1$. So we write $A_2(\Vector{u})=\su\B_2(\Vector{u})$.  
By the last row of $\dD$,  the degrees of $u_0, u_1, u_2, u_3$ in $A_2(\Vector{u})$ are $2, 2, 2, 2$. 
Thus, the degrees of $u_0, u_1, u_2, u_3$ in $\B_2(\Vector{u})$ are $1, 1, 1, 1$, respectively. 
By ({\bf F1}), $A_2$ is homogenous, and we have $\deg(A_2)=\deg(A_3)=2$. So $R_2$ is also homogenous, and $\deg(R_2)=\deg(A_2)-\deg(\su)=1$. 
Then we can assume $\B_2(\Vector{u})=\sum_{k=0}^3C_ku_k$, where 
$C_k\in {\mathbb Q}$.  In order to determine the four coefficients $C_k$, we establish four linear equations by sampling four times.
The correctness of this sampling step is guaranteed by Corollary \ref{cry:ei} (1).
We show below how to do the sampling and establish the first linear equation \eqref{eq:samp1}  in details. The other equations \eqref{eq:samp2}--\eqref{eq:samp4} are similarly obtained. 

Here, we show the steps for the first sampling. For every $u_j$, substitute $u_j= b_j$ into $f_0, \ldots, f_5$, where $b_j$ is a random rational number. For instance, 
we choose $\Vector{b}=(b_0, b_1, b_2, b_3)=(5, 6, 11, 32)$. 
We substitute $u_j=b_j$, and rename the resulting polynomials as $\tilde{f}_0, \ldots, \tilde{f}_5$. 
Note $\tilde{f}_k=f_k(\Vector{b}, \Vector{p}, \Vector{\lambda})$. 
We compute a generator  of $\sqrt{\langle \tilde{f}_0, \ldots, \tilde{f}_{5}\rangle\cap {\mathbb Q}[p_0]}$ and make it monic: 
\[\tilde{g}^{(1)}(p_0)~=~p_0^3-\frac{7}{5}p_0^2+\frac{481}{1458}p_0-\frac{5}{243}.\]
By Corollary \ref{cry:ei} (1), if $\Vector{b}$ is generic in $\C^4$, then $\tilde{g}^{(1)}(p_0)=\frac{\h(\Vector{b},  p_0)}{A_3(\Vector{b})}$. So 
$\coeff(\tilde{g}^{(1)}, p_0^2)=\frac{A_2(\Vector{b})}{A_3(\Vector{b})}$. 
By the discussion above, $A_2(\Vector{u})=\su\sum_{k=0}^3C_ku_k$ and by {\bf Step 2},
 $A_3(\Vector{u})=\su^2$. So we have 
\begin{align}\label{eq:samp1}
-\frac{7}{5}~=~\coeff(\tilde{g}^{(1)}, p_0^2)~=~ \frac{A_2(\Vector{b})}{A_3(\Vector{b})}
~=~\frac{5C_0+6C_1+11C_2+32C_3}{54}.
\end{align}
Similarly, we obtain the other linear equations by samplings:

\noindent
\begin{align}\label{eq:samp2}
-\frac{311}{120}
~=~\frac{11C_0+2C_1+3C_2+8C_3}{24},
\end{align}
\begin{align}\label{eq:samp3}
-\frac{244}{115}~=~
\frac{7C_0+2C_1+5C_2+9C_3}{23},
\end{align}
\begin{align}\label{eq:samp4}
-\frac{181}{110}~=~
\frac{7C_0+3C_1+13C_2+21C_3}{44}.
\end{align}

\noindent
Solve $C_0,\ldots, C_4$ from the $5$ linear equations \eqref{eq:samp1}--\eqref{eq:samp4}, we have
 \[C_0=-\frac{43}{10},\; C_1=-2, \;C_2=-\frac{3}{2},\; C_3=-\frac{4}{5},\] 
 and hence
$A_2(\Vector{u})=-\su(\frac{43}{10}u_0+2u_1+\frac{3}{2}u_2+\frac{4}{5}u_3)$. 
The computational result is consistent with \eqref{eq:ex1h} computed by Gr\"obner bases, if we make $\h$ in \eqref{eq:ex1h} monic with respect to 
$\su^2p_0^3$.  One can similarly interpolate 
$A_0(\Vector{u})$ and $A_1(\Vector{u})$.

\section{Implementation and Computational Results}\label{sec:implementation}

In this section, we explain the implementation details, and compare the timings of Algorithm \ref{interpolation} and existing methods by testing a list of 
interesting algebraic models.

\subsection{Implementation} We first explain the implementation and experimental details. 

\begin{description}
	\item[Software] Algorithm \ref{interpolation} has been implemented in {\tt Maple 2018}, where we use the  {\tt FGb} command {\tt fgb\_gbasis\_elim}
for computing  elimination ideals, for instance,  
in Algorithm \ref{degree}-Lines \ref{elim22}, \ref{elim28},  Algorithm \ref{sampleLC}-Line \ref{elim31} and Algorithm \ref{sample}-Line \ref{elim61}.
{\tt Maple} code and computational results are available online via:
\end{description}
{\footnotesize
\begin{center}
	 \url{https://sites.google.com/site/rootclassification/publications/supplementary-materials/lle2018}.
\end{center}
}
\begin{description}
\item[Hardware and System] We used a 3.2 GHz Intel Core i5 processor (8 GB of RAM) under OS X 10.9.3. 
\item[Testing Models] Models \ref{ex:l1}--\ref{ex:l9} are chosen from the literatures \citep*{SAB2005, DSS2009} and have been tested by both standard elimination method and 
 Algorithm \ref{interpolation}. See the Appendix \ref{appendix} for more details.
\end{description}

\subsection{Computing elimination ideals}
 We have computed the radical elimination ideals $\h$ for Models \ref{ex:l1}--\ref{ex:l9} by standard elimination, \cite[Algorithm 2]{Tang2017} and Algorithm \ref{interpolation}. 
 Table \ref{literatureOld} compares the timings of the three methods.  

 \smallskip
 
 \noindent
{\bf Conclusion from Table  \ref{literatureOld}:}
For smaller models with ML-degree less than $5$, computing Gr\"obner bases directly (standard elimination) is the fastest;
for larger models with ML-degree greater than $5$, Algorithm \ref{interpolation} is the fastest. Particularly, comparing columns ``Interpolation" and ``Algorithm \ref{interpolation}", we see the structure of elimination ideals revealed by Theorem \ref{th:main} (or, Corollary \ref{cry:main}) indeed improves the efficiency significantly.  

\smallskip
 
 \noindent
 {\bf Instruction  for Table \ref{literatureOld}:} 
 \begin{enumerate}
 \item The columns ``$\# p_i$''  and ``ML-Degree'' give the number of probability variables $n$ and ML-degree $N$, respectively.
\item We record the timings of standard elimination in the column ``standard". Standard elimination means to compute the elimination ideal $\langle \Vector{f}\rangle\cap \Q[\Vector{u}, p_0]$ 
by 
running  {\tt FGb} command {\tt fgb\_gbasis\_elim}.  When {\tt FGb} returned no output until we run out the memory,  we record ``$\infty$''.  
\item We record the timings of \cite[Algorithm 2]{Tang2017} and Algorithm \ref{interpolation} in columns ``Interpolation" and ``Algorithm \ref{interpolation}". The italics font timings in columns ``Interpolation" and ``Algorithm \ref{interpolation}" means the computation did not finish in two weeks, but we estimate the sampling timing 
providing a lower bound (see Example \ref{ex:lb}). 
\end{enumerate}

\begin{example}\label{ex:lb}
The italics timings in Table \ref{literatureOld} are the estimated total timings for sampling ({\bf Step 3} in Section \ref{sec:runex}). 
We explain how to estimate these timings by the running example in Section \ref{sec:runex}. 
There are $4$ parameters $u_{i}$ $(i=0, 1, 2, 3)$. We know $\deg(A_2, u_i)$ are all $2$. Also 
by ({\bf F1}) and ({\bf A1}), $A_2$ is homogenous and $\deg(A_2)=\deg(A_2, u_0)=2$. So $A_2$ is a linear combination of $10$ monomials. 
\cite[Algorithm 2]{Tang2017} interpolates $A_2(\Vector{u})$ directly without any structure, so we need to sample $10$ times. 
However, Algorithm \ref{interpolation} interpolates $R_2(\Vector{u})$, which is a factor of $A_2(\Vector{u})$ as shown in the running example,  so we only need to sample 
$4$ times since 
there are $4$ possible monomials in $R_2(\Vector{u})$. 
 We check by {\tt Maple} the timing for doing sample once in {\bf Step 3} is 0.02 second. 
  Then we estimate the timing of sampling 
in \cite[Algorithm 2]{Tang2017} and Algorithm \ref{interpolation} are $0.02 \times 10=2$ (seconds) and $0.02 \times 4=0.08$ (second), respectively.  
\end{example}

\subsection{Computing discriminants}\label{sec:discrim}
Given $\h(\Vector{u}, p_0)$, one straightforward way to get $\mathrm{discr}(\h; p_0)$ is to run {\tt Maple} command {\tt discrim}.  
When $\h$ is large,  we suggest to apply Corollary \ref{cry:discr} since {\tt discrim} might  not be efficient enough, see Example \ref{ex:discr}. 
By the approach described in Example \ref{ex:discr}, we have computed discriminants for Models \ref{ex:l4}, \ref{ex:l6} and \ref{ex:l9}.
We compare the computational timing of our method with \cite[Algorithm 2]{Tang2017} in Table \ref{comparediscr}.

\noindent
{\bf Conclusion from Table  \ref{comparediscr}.}
The new proposed method for computing $\mathrm{discr}(\h; p_0)$ is much more faster than  \cite[Algorithm 2]{Tang2017} for computing $\DD_{J}$, where
$\DD_{J}$ is a factor of   $\mathrm{discr}(\h; p_0)$.

\noindent
 {\bf Instruction  for Table \ref{comparediscr}.} 
 \begin{enumerate}
 \item The columns ``Degree''  and ``Size'' give the total degree of $\mathrm{discr}(\h; p_0)$ and the size of text file, respectively.
\item The italics font timing in the last column ``\cite[Algorithm 2]{Tang2017}" means the computation did not finish in two weeks, but the sampling timing 
can be estimated as shown in \cite[Example 7]{Tang2017}. 
\end{enumerate}

\begin{table}[t]
\small
\centering
\begin{tabular}{|c|c|c|c|c|c|c|} \hline
\multirow{2}{*}{Models}&\multirow{2}{*}{Degree}&\multirow{2}{*}{Size}&
\multicolumn{3}{|c|}{Our Method}&\multirow{2}{*}{\cite[Algorithm 2]{Tang2017}}\\
\cline{4-6} 
 &&&$\h$& $\mathrm{discr}(\h;p_0)$ &Total&\\ \hline
  Model \ref{ex:l4} &110 & 7.5 MB&782.676 s &0.027 s &{\bf 783} s&805 s\\ \hline
      Model \ref{ex:l6} &342&\textcolor{red}{$>${\it 32 GB}}  &14 d & 11.379 s&{\bf 14} d &\textcolor{red}{$>${\it 13374 d}}\\ \hline
            Model \ref{ex:l9} &176&  8.68 GB & 2 d &81.015 s&{\bf 2} d&\textcolor{red}{$>${\it  454833 d}}\\ \hline
\end{tabular}
\smallskip 
\caption{Runtimes for computing discriminants $\mathrm{discr}(\h;p_0)$ (s: seconds; d: days). }
\label{comparediscr}
\end{table}

\bibliographystyle{alpha}
\bibliography{rs}

\appendix

\section{Testing Models in Table \ref{literatureOld}}\label{appendix}

\footnotesize

\begin{model}\citep*[Random Censoring Model]{DSS2009}\label{ex:l1}
\[2p_0p_1p_2 + p_1^2p_2 + p_1p_2^2 - p_0^2p_{12} + p_1p_2p_{12}=0, \;\;\; p_0 + p_1 + p_2 + p_{12} = 1\]
\end{model}

\begin{model}\citep*[$3\times 3$ Zero-Diagonal Matrix]{EJ2014}\label{ex:l2}
 {\begin{align*}
\det \left[
\begin{array}{ccc}   
    0&    p_{12}    & p_{13} \\   
    p_{21} &    0 & p_{23}\\   
    p_{31} &   p_{32} &  0
\end{array}
\right]=0,\;\;\;p_{12} + p_{13} + p_{21} + p_{23} + p_{31} + p_{32} =1
\end{align*}}
\end{model}

\begin{model}\citep*[Grassmannian of $2$-planes in ${\mathbb C}^4$]{SAB2005, EJ2014}\label{ex:l3}
\[p_{12}p_{34}-p_{13}p_{24}+p_{14}p_{23}=0, \;\;\; p_{12} + p_{13} + p_{14} + p_{23} + p_{24} + p_{34} =1\]
\end{model}
\begin{model}\citep*[$3\times 3$ Symmetric Matrix]{SAB2005}\label{ex:l4}
{\begin{equation*}
\det\left[
\begin{array}{cccc}   
    2p_{11} &    p_{12}    & p_{13} \\   
    p_{12} &    2p_{22}   & p_{23}\\   
    p_{13} & p_{23} & 2p_{33} 
\end{array}
\right]=0, \;\;\; p_{11} + p_{12} + p_{13} + p_{22} + p_{23} + p_{33} =1
\end{equation*}}
\end{model}

\begin{model}\citep*[Bernoulli $3\times 3$ Coin]{SAB2005}\label{ex:l5}
{\footnotesize \begin{equation*}
\det\left[
\begin{array}{cccc}   
    12p_{0} &    3p_{1}    & 2p_{2} \\   
    3p_{1} &    2p_{2}   & 3p_{3}\\   
    2p_{2} & 3p_{3} & 12p_{4} 
\end{array}
\right]=0,\;\;\; p_{0} + p_{1} + p_{2} + p_{3} + p_{4}  =1
\end{equation*}}
\end{model}

\begin{model}\citep*[$3\times 3$ Matrix]{SAB2005}\label{ex:l6}
{\footnotesize \begin{equation*}
\det\left[
\begin{array}{cccc}   
    p_{00} &    p_{01}    & p_{02} \\   
    p_{10} &    p_{11}   & p_{12}\\   
    p_{20} & p_{21} &  p_{22} 
\end{array}
\right]=0, \;\;\; p_{00} + p_{01} + p_{02} + p_{10} + p_{11} + p_{12} +  p_{20} + p_{21} + p_{22} =1
\end{equation*}}
\end{model}


\begin{model}\citep*[Juke-Cantor Model, Example 18]{SAB2005}\label{ex:l7}
\[q_{000}q_{111}^2 - q_{011} q_{101} q_{110}=0, \;\;\; p_{123} + p_{dis} + p_{12} + p_{13} + p_{23}=1\]
where\\
$q_{111} = p_{123} + \frac{p_{dis}}{3} - \frac{p_{12}}{3} - \frac{p_{13}}{3} - \frac{p_{23}}{3}$, 
$q_{110} = p_{123} -  \frac{p_{dis}}{3} + p_{12} - \frac{p_{13}}{3} - \frac{p_{23}}{3}$,\\
$q_{101} = p_{123} -  \frac{p_{dis}}{3} - \frac{p_{12}}{3} + p_{13} - \frac{p_{23}}{3}$, 
$q_{011} = p_{123} -  \frac{p_{dis}}{3} - \frac{p_{12}}{3} - \frac{p_{13}}{3} + p_{23}$,\\
$q_{000} = p_{123} + p_{dis} + p_{12} + p_{13} + p_{23}$.
\end{model}

\begin{model}\citep*[Example 15]{SAB2005}\label{ex:l8}
\[q_2q_7-q_1q_8=0, \;\;\;q_3q_6-q_5q_4=0, \;\;\;
p_1 + p_2 + p_3 + p_4 + p_5 + p_6 + p_7 + p_8=1\]
where\\
$q_1 = p_1 + p_2 + p_3 + p_4 + p_5 + p_6 + p_7 + p_8$, 
$q_2 = p_1 - p_2 + p_3 - p_4 + p_5 - p_6 + p_7 - p_8$,\\
$q_3 =  p_1 + p_2 - p_3 - p_4 + p_5 + p_6 - p_7 - p_8$,
$q_4 =  p_1 - p_2 - p_3 + p_4 + p_5 - p_6 - p_7 + p_8$,\\
$q_5 =  p_1 + p_2 + p_3 + p_4 - p_5 - p_6 - p_7 - p_8$,
$q_6 =  p_1 - p_2 + p_3 - p_4 - p_5 + p_6 - p_7 + p_8$,\\
$q_7 =  p_1 + p_2 - p_3 - p_4 - p_5 - p_6 + p_7 + p_8$,
$q_8 =  p_1 - p_2 - p_3 + p_4 - p_5 + p_6 + p_7 - p_8$.
\end{model}

\begin{model}\citep*[$P_{comb}$, Example 15]{SAB2005}\label{ex:l9}
\[q_3-q_5,  \;\;\;q_2-q_5, \;\;\; q_4-q_6, \;\;\;q_5q_7-q_1q_8=0, \;\;\;
p_1 + p_2 + p_3 + p_4 + p_5 + p_6 + p_7 + p_8=1\]
where\\
$q_1 = p_1 + p_2 + p_3 + p_4 + p_5 + p_6 + p_7 + p_8$, 
$q_2 = p_1 - p_2 + p_3 - p_4 + p_5 - p_6 + p_7 - p_8$,\\
$q_3 =  p_1 + p_2 - p_3 - p_4 + p_5 + p_6 - p_7 - p_8$,
$q_4 =  p_1 - p_2 - p_3 + p_4 + p_5 - p_6 - p_7 + p_8$,\\
$q_5 =  p_1 + p_2 + p_3 + p_4 - p_5 - p_6 - p_7 - p_8$,
$q_6 =  p_1 - p_2 + p_3 - p_4 - p_5 + p_6 - p_7 + p_8$,\\
$q_7 =  p_1 + p_2 - p_3 - p_4 - p_5 - p_6 + p_7 + p_8$,
$q_8 =  p_1 - p_2 - p_3 + p_4 - p_5 + p_6 + p_7 - p_8$.
\end{model}


\end{document}